\documentclass[11pt]{article}
\usepackage[utf8x]{inputenc}
\usepackage[OT2,T1]{fontenc}
\makeatletter
\newcommand{\mylabel}[2]{#2\def\@currentlabel{#2}\label{#1}}
\makeatother
\usepackage{braket}
\usepackage{amsfonts,amssymb}
\usepackage{color}
\usepackage{mathtools}
\usepackage{amsthm, amssymb, latexsym, amsmath, dsfont
}

\usepackage[margin=1.2 in]{geometry}
\usepackage[ngerman,english]{babel}
\usepackage{mathabx}
\usepackage{booktabs}
\usepackage{siunitx}
\usepackage{natbib}
\usepackage{hyperref}
\makeatletter
\def\eqref{\@ifstar\@eqref\@@eqref}
\def\@eqref#1{\textup{\tagform@{\ref*{#1}}}}
\def\@@eqref#1{\textup{\tagform@{\ref{#1}}}}
\makeatother

\newtheorem{theorem}{Theorem}[section]

\newtheorem{lemma}{Lemma}[section]

\newtheorem{proposition}{Proposition}[section]

\newtheorem{definition}{Definition}[section]

\newtheorem{corollary}{Corollary}[section]

\theoremstyle{remark}
\newtheorem{remark}{Remark}[section]

\usepackage[capitalise,nameinlink,noabbrev]{cleveref}
\setcitestyle{super,comma,square,numbers}

\usepackage{braket,amsfonts,amssymb}
\usepackage[caption=false]{subfig}
\usepackage{array}
\usepackage[shortlabels]{enumitem}
\usepackage{pgfplots}
\usepackage{graphicx,epstopdf}
\usepackage{tikz}
\usetikzlibrary{shapes.geometric, arrows}
\tikzstyle{startstop} = []
\tikzstyle{process} = [rectangle, minimum width=2.5cm, minimum height=0.8cm, text centered, draw=black, fill=white!20]
\tikzstyle{arrow} = [thick,->,>=stealth]
\usepackage{graphicx,epstopdf}
\usepackage{booktabs}
\usepackage{amsopn}
\usepackage{mathrsfs}  
\usepackage[mathscr]{euscript}
\usepackage{bm,nicefrac}
\usepackage{extarrows}

\def\vA{\mathbf{A}}

\def\va{\mathbf{a}}
\def\vbb{\mathbf{B}}
\def\vbbb{\mathbf{b}}
\def\vC{\mathbf{C}}

\def\vD{{\mathbf{D}}}

\def\vdd{\mathbf{d}}

\def\vi{\mathbf{i}}

\def\vm{\pmb{\mathscr{M}}}

\def\vN{N}
\def\vp{\beta}
\def\vP{\mathbf{Pr}}

\def\vt{\tau}
\def\vU{\mathbf{U}}
\def\vu{\mathbf{u}}

\def\vv{\mathbf{v}}
\def\vW{\mathbf{W}}

\def\vM{\mathbf{M}}

\def\vX{\pmb{\mathscr{X}}}
\def\vx{\mathbf{X}}
\def\vxx{\mathbf{x}}
\def\vY{\mathbf{Y}}
\def\vy{\mathbf{y}}
\def\vZ{\mathbf{Z}}

\def\vPsi{\bm{\Psi}}
\def\vPhi{\bm{\Phi}}

\def\vF{\mathbf{F}}

\def\vS{\mathbf{S}}

\def\Id{\mathbf{I}}

\newcommand{\Real}{\mathbb{R}}
\newcommand{\eps}{\varepsilon}
\newcommand{\nrmsqr}[1]{\| #1 \|_2^2}
\newcommand{\qtext}[1]{\quad\text{#1}\quad}
\newcommand{\sE}{\mathcal{E}}

\begin{document}

\title{Faster Johnson-Lindenstrauss Transforms via Kronecker Products}
	
\author{
Ruhui Jin\thanks{Department of Mathematics, University of Texas, Austin, TX},\quad
Tamara G. Kolda \thanks{Sandia National Laboratories, Livermore, CA},\quad
and Rachel Ward\footnotemark[1]}
	
	\bigskip
	
	\bigskip
	
	\bigskip

	\bigskip
	\bigskip

\maketitle

\begin{abstract}
The Kronecker product is an important matrix operation with a wide range of applications in signal processing, graph theory, quantum computing and deep learning.  In this work, we introduce a generalization of the fast Johnson-Lindenstrauss projection for embedding vectors with Kronecker product structure, the \emph{Kronecker fast Johnson-Lindenstrauss transform} (KFJLT).  
The KFJLT reduces the embedding cost by an exponential factor of the standard fast Johnson-Lindenstrauss transform (FJLT)'s cost when applied to vectors with Kronecker structure, by avoiding explicitly forming the full Kronecker products.  
We prove that this computational gain comes with only a small price in embedding power: consider a finite set of $p$ points in a tensor product of $d$ constituent Euclidean spaces $\bigotimes_{k=d}^{1}\Real^{n_k}$, and let $N = \prod_{k=1}^{d}n_k$.  With high probability, a random KFJLT matrix of dimension  $m \times N$ embeds the set of points up to multiplicative distortion $(1\pm \eps)$ provided $m \gtrsim \varepsilon^{-2} \, \log^{2d - 1} (p) \, \log N$.   
We conclude by describing a direct application of the KFJLT to the efficient solution of large-scale Kronecker-structured least squares problems for fitting the CP tensor decomposition.
\end{abstract}


\section{Introduction}
Dimensionality reduction is commonly used in data
analysis to project high-dimensional data onto a
lower-dimensional space while preserving as much information as
possible.
The powerful \emph{Johnson-Lindenstrauss lemma} proves the existence of a class of linear maps which provide low-distortion embeddings of an arbitrary number of points from high-dimensional Euclidean space into an arbitrarily lower dimensional space \cite{JL84,DG03}. 
A (distributional) Johnson-Lindenstrauss transform (JLT) is a random linear map which provides such an embedding with high probability, and a \emph{fast} JL transform (FJLT) exploiting fast matrix-vector multiplies of the FFT significantly reduces the complexity of the embedding with only a minor increase in the embedding dimension \cite{AC06,AC09,AL10}. We consider the dimensionality reduction problem for high-dimensional subspaces with \emph{structure},
specifically, subspaces corresponding to a tensor product of lower-dimensional Euclidean spaces.
In this case, we can dramatically reduce the embedding complexity with only a small increase in the embedding dimension (see \cref{savings}). 

\begin{figure}[tbhp]
\label{savings}
\centering 
\subfloat[embedding time]{\includegraphics[width=0.5\textwidth]{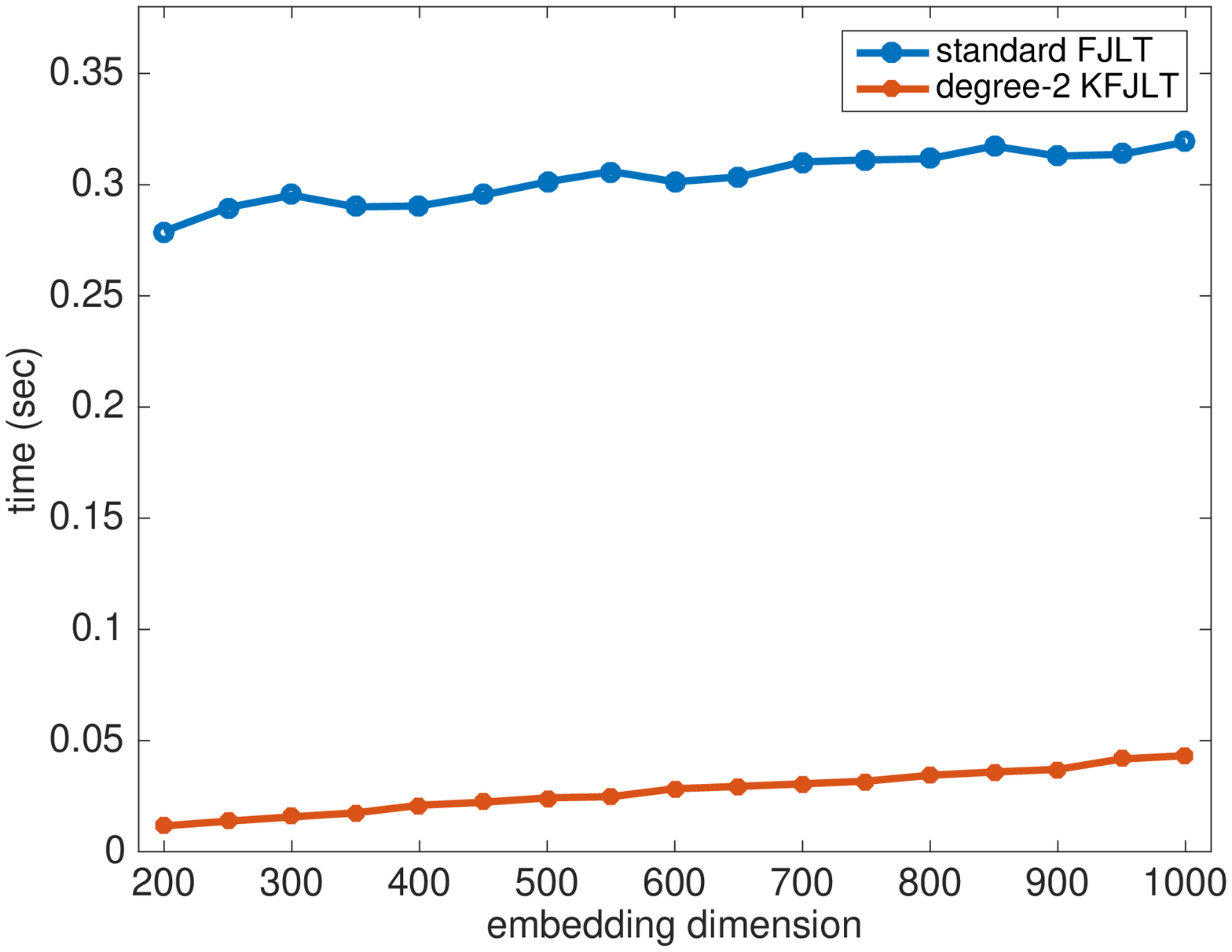}}~~
\subfloat[embedding distortion]{\includegraphics[width=0.5\textwidth]{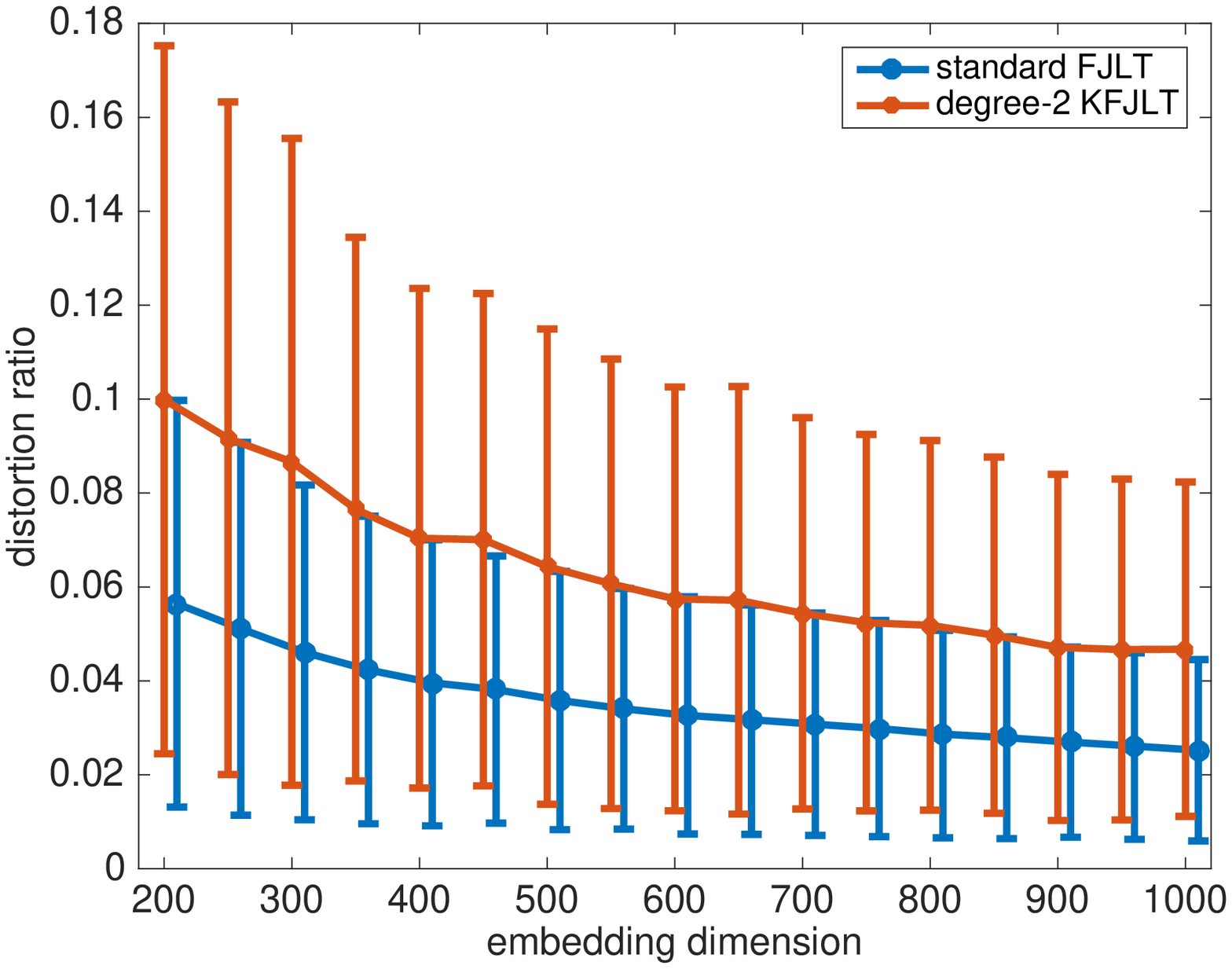}}
\caption{Comparison of the embedding time and distortion between a standard FJLT and a Kronecker FJLT on vectors with Kronecker structure $\Real^{125} \otimes \Real^{125}$ using $\mathrm{MATLAB~ R}2015\mathrm{a}$ fft() and Tensor Toolbox $v3.1$ \cite{TTB} on a standard MacBook Pro $2016$ with $16$ GB of memory. The left figure shows the total embedding time for all $1000$ appropriately structured Kronecker vectors. The right shows the average embedding distortion with vertical lines to indicate the standard deviation over the $1000$ trials. Each component vector consist of normally distributed elements. In the right figure, the x-axis of the standard FJLT plot is shifted by $+10$ to show the error bars clearly.}
\end{figure}

\subsection{Review of JLT and FJLT}
\label{sec:review-jlt-fjlt}

We briefly review JLT and FJLT.
Suppose we have a set $\sE \subset \Real^N$ of $p$ points.
A JLT is a (random) linear map $\vPhi$ from $\Real^N$ to $\Real^m$ with $m$ as small as $m = O(\varepsilon^{-2} \, \log p)$ in the optimal scaling \cite{JL84, LN17} such that with high probability with respect to the draw of $\vPhi$, the transformed points have at most $(1 \pm \eps)$ multiplicative distortion, i.e.,
\begin{equation}
  \label{eq:distorition}
  (1-\eps) \nrmsqr{\vxx} \leq \nrmsqr{\vPhi \vxx} \leq (1+\eps) \nrmsqr{\vxx}
  \qtext{for all} \vxx \in \sE,
\end{equation}
or, more concisely, we denote this as
\begin{equation}
\|\vPhi \vxx\|_2^2 = (1 \pm \eps)\|\vxx\|_2^2  \qtext{for all} \vxx \in \sE.
\end{equation}

The Gaussian testing matrix 
\begin{equation}
  \label{eq:jlt_example}
  \vPhi \in \Real^{m_{\text{g}} \times N}
  \qtext{with}
  m_\text{g}  = O(\varepsilon^{-2}\, \log p)
  \qtext{and  i.i.d.}
  \phi_{ij} \sim \mathcal{N}(0,1/m)
\end{equation}
is a particular JLT which achieves the optimally small $m$ as a function of $p$ and $\varepsilon$.
Although this is a powerful result, the cost of the per-point transformation $\vxx \rightarrow \vPhi \vxx$ is $O(m_{\text{g}}N)$, limiting applicability. 
To reduce the cost, the \emph{fast} JLT (FJLT) was introduced to employ fast matrix-vector multiplication \cite{AC06,AC09,AL10}.
An example FJLT is of the form
\begin{equation}
  \begin{gathered}
  \label{eq:fjlt_example}
  \vPhi = \sqrt{\frac{N}{m_{\text{f}}}}\,\vS \vF \vD_{\xi} \in \mathbb{C}^{m_{\text{f}}\times N} 
  \qtext{with}
  m_\text{f}=O\Big(\varepsilon^{-2}\, \log p \, \log^4(\log p)  \, \log N \Big) {\text{\cite{KW11,HR16}}}, \\
  \begin{split}
    \vS \in \Real^{m_{\text{f}} \times N} &= \text{$m_\text{f}$ random rows of the $N \times N$ identity matrix},\\
    \vF \in \mathbb{C}^{N \times N} &= \text{(unitary) discrete Fourier transform of dimension $N$},\\
    \vD_{\xi} \in \Real^{N \times N} &= \text{diagonal matrix with diagonal entries $\xi(i)$, and}\\
\xi\in \Real^N &= \text{vector with independent random entries drawn uniformly from $\{-1, +1\}$}.
  \end{split}    
  \end{gathered}
\end{equation}
On the one hand, the embedding dimension, $m_\text{f}$, in the FJLT is increased by a factor of $\log^4(\log p)  \, \log N$ as compared to the optimal JLT.  On the other hand, the per-point transformation cost is reduced from $O(m_{\text{g}}N)$ to $O(N \log N + m_{\text{f}})$.

\subsection{Our contribution: Kronecker FJLT}
The Kronecker product is an important matrix operation with a wide range of applications in signal processing \cite{DR93, DB12}, graph theory \cite{LCKFG10}, quantum computing \cite{FW98}, deep learning \cite{MG15}, and inverse problems \cite{JCCL19,CLNW19}, just to name a few.

In this work, we propose a \emph{Kronecker} FJLT (KFJLT) of the following form
\begin{equation}
  \begin{gathered}
  \label{eq:kfjlt_example}
   \begin{split}
     \vPhi &= \sqrt{\frac{N}{m_{\text{kron}}}}\,\vS \bigotimes_{k=d}^1 ( \vF_k \vD_{\xi_k} ) \in \mathbb{C}^{m_{\text{kron}} \times N},\\
 \qtext{with}
  m_\text{kron}&=O\Big(\varepsilon^{-2} \, \log^{2d-1} (p) \, \log^4(\log p) \, \log N \Big),
  \qtext{and} N = \prod_{k=1}^d n_k,\\
     \vS \in \Real^{m_{\text{kron}} \times N} &= \text{$m_\text{kron}$ random rows of the $N \times N$ identity matrix},\\
    \vF_{k} \in \mathbb{C}^{n_k \times n_k} &= \text{(unitary) discrete Fourier transform of dimension $n_k$,}\\
    \vD_{\xi_k} \in \Real^{n_k \times n_k} &= \text{diagonal matrix with diagonal entries $\xi_k(i)$, and}\\
    \xi_{k}\in \Real^{n_k} &= \text{vector with independent random entries drawn uniformly from $\{-1, +1\}$}.
  \end{split}    
  \end{gathered}
\end{equation}
The $\vS$ matrix is unchanged, but $\vF\vD_\xi$ has been replaced by a Kronecker product. We call $d$ the \emph{degree} of the KFJLT.

Suppose each vector $\vxx \in \sE$ belongs to a tensor product space: 
\begin{displaymath}
  \vxx = \bigotimes_{k=d}^1 \vxx_k \in \Real^N
  \text{ where } \vxx_k \in \Real^{n_k},
  \qtext{i.e.,}
  x(i) = \prod_{k=1}^d x_k(i_k)
  \text{ where } i = 1+\sum_{k=1}^d (i_k-1)\prod_{\ell=1}^{k -1} n_\ell.
\end{displaymath}
The KFJLT in this scenario reduces the per-point transformation cost to \\
$O\left( \sum_{k=1}^d n_k \log n_k + m_{\text{kron}} \right)$. As compared to the FJLT, the necessary embedding dimension has increased by a factor of $\log^{2d-2} (p) $. When $d=1$, the KFJLT reduces to the standard FJLT. This idea was proposed in the context of matrix sketching for the least squares problems in fitting the CANDECOMP/PARAFAC (CP) tensor decomposition \cite{BBK18}; however, there was no proof such a transform was a JLT. In this work, we prove that this is a JLT and that the embedding dimension is only slightly worse than in the FJLT case.

\subsection{Related work}
Sun et~al.~\cite{SGTU18} proposed a related tensor-product embedding construction called the tensor random projection (TRP). The TRP is a low-memory framework for random maps formed by a row-wise Kronecker product of common embedding matrices; for example, Gaussian testing matrices and sparse random projections. The authors provide theoretical analysis for the case of the component random maps being two Gaussian matrices. 
The TRP idea was previously used in \cite{BBB15} for tensor interpolative decomposition, but without any theoretical guarantees.   Our theoretical embedding results are favorable compared to those in \cite{SGTU18} in several key aspects: our embedding bound applies to fast JLTs which support fast matrix multiplications, our embedding bound holds for the general degree-$d$ case while they only consider the degree-$2$ case, and even in the degree-$2$ case, the necessary embedding dimension we provide is $O(\eps^{-2} \, \log^3 (p))$, which is significantly smaller than the $O(\eps^{-2} \, \log^{8} (p))$ proved in \cite{SGTU18}.

More peripherally, \textsc{TensorSketch} developed in \cite{PP13}, \cite{PP13} is a popular dimension reduction technique utilizing FFT and fast convolution to recover the Kronecker product of \textsc{CountSketch}ed \cite{CCC04} vectors, but TensorSketch is not a JLT.  Diao et al. \cite{DSSW18} extends the applications of \textsc{TensorSketch} to accelerating Kronecker regression problems by creating oblivious subspace embedding (OSE) \cite{ANW14} without explicitly forming Kronecker products for coefficient matrices.

Simultaneous work \cite{KPVWZ19} develops a tensor sketching method for approximating structured polynomial kernels. One part of the paper introduces the so-called TensorSHRT, which is similar to the KFJLT in our case. However, our result, which  main distinction is the authors only prove bounds particular to the subsampled Hadamard matrix, while our result applies more broadly to a general class of column-randomized RIP matrices including the subsampled Hadamard matrix, the subsampled FFT we focus on,  and bounded orthonormal system constructions more generally  \cite{HR16}.   The embedding dimension in  \cite{KPVWZ19} required to satisfy the Strong JL property is $m \sim \eps^{-2} \log(1/\eta)^{d+1}$, which is better than our bound $m \sim \eps^{-2} \log(1/\eta)^{2d-1}$ once $d \geq 3$.  An interesting open question is whether the bound $m \sim \eps^{-2} \log(1/(\eta))^{d+1}$ holds for the general RIP-based tensor JL constructions considered here.

Another simultaneous work \cite{MB19} on embedding properties of the KFJLT uses a different proof to show that the KFJLT is a JLT transform once $m \gtrsim \varepsilon^{-2} \,\log (p)\,\prod_{k=1}^d \log (n_k\, p) \, \log N,$ compared to ours \\
($m \gtrsim \varepsilon^{-2} \, \log^{2d - 1} (d\, p) \, \log N$). However, their JLT embedding results hold only for Kronecker-structured vectors, while our main result holds for generic vectors. As the main application of the KFJLT -- in accelerating the algorithm CPRAND-MIX for tensor least squares fitting -- involves application of the KFJLT to arbitrary (not Kronecker-structured vectors) vectors as part of the one -time pre-computation step before applying Alternating Least Squares (see CPRAND-MIX for more details). Our result can apply to arbitrary vectors directly, while their approach requires a more complicated argument to achieve the least-squares result.

Work building upon this work \cite{INRZ19} proposes a new type of mode-wise embedding construction for the tensor subspace embedding. In particular, one of the main result in their paper is built on our result Theorem 2.1 with a coarser covering net argument. As a consequence, they show an improved subspace embedding property for the KFJLT given certain coherence assumptions. 

\subsection{Organization of the paper}
The paper is organized as follows. Section $2$ presents our main result concerning the embedding dimension bound of KFJLT (\cref{main}), and its theoretical implications for the application of KFJLT in an algorithm for CP tensor decomposition  (\cref{sample size}). The derivation of the latter result relies on the result from \cref{main}. 
We present the proof of \cref{main} in Section $3$ using the concentration property of a matrix with a so-called restricted isometry property and randomized column signs from the Kronecker Rademacher vector given in \cref{analog} and the best known bound of the RIP matrices (\cref{RIPsharp} from \cite{HR16}). The proof of \cref{analog} via the probability bound shown in \cref{prob bound} is also contained in this section.
Section $4$ shows the proof of \cref{prob bound}, which simply follows a more generalized result given in \cref{generalized prob bound} via induction on the degree $d$. 
We conclude by presenting the numerical results in Section $5$. 
The proof of an intermediate result \cref{modified} (used for \cref{generalized prob bound}) is shown in \cref{modified proof}. We briefly introduce the randomized method of fitting CP tensor model via the alternating least squares in \cref{CPRAND-MIX}.

\subsection{Notation}
We denote by $\|\,\|_{2}$ and $\|\,\|_{\infty}$ respectively the $\ell_2$ and $\ell_\infty$ norms of a vector, and $\|\,\|, ~\|\,\|_{F}$  respectively the spectral and Frobenius norm of a matrix. We use the Euler script uppercase letter $\vX$ to denote tensors, the Roman script uppercase letter $\vx$ to denote matrices, the Roman script lowercase letter $\vxx$ to denote vectors, and simple lowercase letter $x$ to denote a scalar entry. We use the capital letter $I$ for index sets and lowercase letter $i$ for single indices. A \emph{Rademacher vector} $\xi \in \mathbb{R}^N$ refers to a random vector whose entries are independent random variables $\xi(i) = \{+1 ~\text{with probability}~ 1/2, -1 ~\text{with probability} ~ 1/2\}$ for $i\in [N]$. We write $\Id_N \in \Real^{N \times N}$ as the $N$ by $N$ identity matrix. For a vector $\vxx \in \Real^N$, we use $\vD_{\vxx}\in \Real^{N\times N}$ to refer to the diagonal matrix satisfying $\vD_{\vxx} (i,i) = \vxx(i)$ for $i\in [N]$.

\section{Main results for KFJLT}
The major part of our work is analyzing the vector embedding property of the KFJLT in order to provide a theoretical bound for the embedding dimension. 
\begin{theorem}
\label{main}
Fix $d \geq 1$ and $\varepsilon, \eta \in (0,1)$.  Fix integers $n_1, n_2, \dots, n_d$ and $N = \prod_{k=1}^d n_k$. Consider a finite set $\sE\subset \mathbb{R}^N$ of cardinality $|\sE|=p \geq N$.  Suppose the KFJLT $\vPhi \in \mathbb{C}^{m_{\text{kron}} \times N}$ defined in \cref{eq:kfjlt_example} has embedding dimension 
\begin{equation}
\label{embedding dimension}
m_{\text{kron}} \geq C\,\left[\varepsilon^{-2}\, \log^{2d - 1} \left( \frac{dp}{\eta} \right) \, \log^4 \left( \frac{\log^d(\frac{dp}{\eta})}{\eps}\right) \,\log N \right].
\end{equation}
Then
\begin{equation*}
\vP\Big(\| \vPhi \vxx\| _{2}^{2} = (1\pm\varepsilon)\| \vxx\| _{2}^{2}, ~~\forall \vxx \in \sE\Big) \geq 1-\eta-2^{-\Omega(\log N)}.
\end{equation*}
Above, $C > 0$ in \cref{embedding dimension} is a universal constant.
\end{theorem}

\begin{remark}
\label{d=1}
For $d=1$, the embedding $\vPhi$ reduces to the standard FJLT corresponding to a random subsampled DFT matrix with randomized column signs.  In this case, the results of \cref{main} are already known, see \cite{KW11}, and stated above for completeness.  The result for $d \geq 2$ is proved in this paper. 
\end{remark}

\begin{remark}
The main theorem holds not only for subsampled discrete Fourier matrices, but more broadly to any matrix satisfying the restricted isometry property (see \cref{rip}) -- including  subsampled Hadamard transforms and DCT matrices -- with randomized column signs.  We focus on the FFT-based fast JL transform for clarity of presentation.
\end{remark}

\begin{remark}
In the KFJLT construction, the randomness in the embedding construction decreases as the degree $d$ increases. Specifically, the Kronecker product of independent random matrices $\bigotimes_{k=d}^{1} \vD_{\xi_{k}}$ consists of $\sum_{k=1}^{d} n_{k}$ bits, compared to $N = \prod_{k=1}^{d} n_{k}$ bits which would be used to construct a standard FJLT. This reduction in randomness is the source of the additional factor of $\log^{2d-2} (p)$ in the number of measurements $m$ required to achieve the quality of approximation compared to the standard FJLT. While we suspect that this additional factor may be pessimistic, some loss of embedding power is necessary with increasing degree $d$.  This is explored numerically in \cref{savings} and \cref{numerics}.
\end{remark}

\begin{remark}
The embedding dimension of the KFJLT stated in the abstract and \cref{eq:kfjlt_example} are simplified to omit constants $d, \eta, \eps$ in the logarithmic term of \cref{embedding dimension}.
\end{remark}

\subsection{Preliminaries}
To clearly illustrate our motivation, we first introduce the linear algebra background.
Given matrices $\vx\in\mathbb{R}^{m\times n}$ and $\vY\in\mathbb{R}^{m'\times n'}$, the Kronecker product of $\vx$ and $\vY$ is defined as
\begin{equation}
\vx \otimes \vY= \left[                 
  \begin{array}{cccc}   
   x(1,1)\vY&x(1,2)\vY &\dots&x(1,n)\vY \\  
    \vdots&\vdots &\ddots&\vdots\\
     x(m,1)\vY&x(m,2)\vY &\dots&x(m,n)\vY \\  
  \end{array}
\right]\in\mathbb{R}^{mm'\times nn'}.
\end{equation}
We frequently use the distributive property of the Kronecker product:
\begin{equation}
\vW\vx\otimes \vY\vZ = (\vW\otimes \vY)(\vx\otimes \vZ).
\end{equation}
Another important property is that the Kronecker product of unitary matrices is a (higher-dimensional) unitary matrix.

\subsection{Cost savings when applied to Kronecker vectors}
Although \cref{main} concerns the general embedding property of the KFJLT $\vPhi$, the embedding is particularly efficient when considered as an operator $\vPhi: \bigotimes_{k=d}^{1} \mathbb{C}^{n_{k}} \rightarrow \mathbb{C}^m$ applied to vectors $\vxx =  \bigotimes_{k=d}^{1} \vxx_k \in \bigotimes_{k=d}^{1} \Real^{n_k}$ with Kronecker product structure matches that of the embedding matrix.  In this setting, the Kronecker mixing on $\bigotimes_{k=d}^{1} \vxx_k$ is equivalent to imposing the mixing operation respectively on each component vector $\vxx_k$, and reduces the mixing cost to a much smaller scale.  As the Kronecker structure of the embedded vector is maintained after the mixing, we are able to start from the sampled elements and trace back to find its forming components based on the invertible linear transformation of indices. This strategy restricts the computation objects to only the sampled ones and saves significant amount of floating point operations and memory cost, compared to conventional embedding methods. See \cref{embedding cost} for the comparison in cost between the standard and Kronecker FJLT on Kronecker vectors.

\begin{table}[tbhp]
{\footnotesize
  \caption{{\bf Embedding cost on Kronecker vectors.} (Here $N = \prod_{k=1}^{d} n_k$.)}
\label{embedding cost}
\begin{center}
  \begin{tabular}{ccccc} 
  \toprule[1pt]
 & \bf Forming Kronecker product & \bf Mixing & \bf Sampling \\[1mm]
\bf FJLT & ${N}$ & ${O\left(N\,\log N\right)+ N}$ & ${m_{\text{f}}}$\\[1.5mm]
\bf KFJLT  & none &${O\left(\sum_{k=1}^{d}n_{k}\log n_{k} \right)+\sum_{k=1}^{d}n_{k}}$ & ${d\,m_{\text{kron}}}$\\ 
 \bottomrule[1pt]
  \end{tabular}
\end{center}
}
\end{table}

Note that we treat the construction degree $d$ as a constant in the complexity.
 
\subsection{Applications to CP tensor decomposition}
\label{tensor}

The study of multiway arrays, \emph{tensors}, has been an active research area in large-scale data analysis, because it is a natural algebraic representation for multidimensional data models. 

The KFJLT technique has been applied as a sketching strategy in a randomized algorithm: CPRAND-MIX for CP tensor decomposition. At each iteration, the \emph{alternating least squares} (CP-ALS) problem for fitting a rank-$R$ tensor model solves a problem of the form:

\begin{equation}
\label{cpals}
\min_{\vx\in \Real^{R \times n}}\| \vA\vx - \vbb \|_{F},
\end{equation}
where 
$\va(i)$ and $\vbbb(j)$ are respectively the $i$-th column of $\vA$ and the $j$-th column of $\vbb$, for $i \in [R]$ and $j \in [n]$.
Each column $\va(i)$ has the Kronecker structure:
\begin{align*}
\va(i) = \bigotimes\limits_{k=d}^1 \va_k(i) \in \bigotimes_{k=d}^1\Real^{n_k}\subset \Real^N.
\end{align*}
We refer the readers to \cref{CPRAND-MIX} and \cite{BBK18} for more details. 

This least squares problem is a candidate for the KFJLT sketching approach in \cref{eq:kfjlt_example}. \Cref{main} demonstrates that KFJLT is a low-distortion embedding for a fixed set of points with high probability. Combining with standard covering arguments, we can also derive a bound for the power of KFJLT in embedding an entire subspace of points (rather than a finite set of points), and through this, provide a theoretical guarantee for the sample size of CPRAND-MIX. 

Note that the matrix least squares \cref{cpals} can be factored into $n$ separate vector least squares on each column $\vbbb(j)$: $\min_{\vxx(j)\in \Real^{R}}\| \vA\vxx(j)- \vbbb(j) \|_2$. Thus, without loss of generality, we simplify the problem setting and focus on the vector-based least squares. 

\begin{proposition}
\label{sample size}
For the coefficient matrix $\vA\in\Real^{N \times r}$ in \cref{cpals} and a fixed vector $\vbbb\in \Real^N$, consider the problem: $\min_{\vxx\in\Real^R} \|\vA\vxx-\vbbb\|_2$. Denote $\text{rank}(\vA) = r \leq R.$ Fix $\eps, \eta \in (0,1)$ such that $N \lesssim 1/\eps^r$ and integer $r \geq 2$. Then a degree-$d$ KFJLT $\vPhi = (\sqrt{N/m})\,\vS\bigotimes_{k=d}^{1}(\vF_{k}\vD_{\xi_{k}}) \in \mathbb{C}^{m \times N}$ with 
\begin{equation}
\label{sample size term}
m = O\left(\eps^{-1}\,r^{2d} \,\log^{2d-1}( \frac{r}{\eps}) \, \log^4\left(\frac{r}{\eps}\,\log( \frac{r}{\eps})\right)\, \log N\right)
\end{equation}
uniformly sampled rows is sufficient to output
\begin{equation}
\label{sketched solution}
\hat{\vxx} = \arg\min_{\vxx\in \Real^{R}}\| \vPhi\vA\vxx - \vPhi\vbbb \|_2,
\end{equation}
such that $$\vP\Big(\| \vA\hat{\vxx}- \vbbb \|_2  = (1\pm O(\eps)) \min_{\vxx\in \Real^R}\| \vA\vxx - \vbbb \|_2\Big) \geq 1-\eta-2^{-\Omega(\log N)}.$$ 
\end{proposition}

The proposition is proved using the following corollary of our main result \cref{main}:
\begin{corollary}
\label{subspace embedding}
Denote $\vU \in \Real^{N\times r}$ as the orthogonal basis for $\vA\in\Real^{N \times R}$ with $\text{rank}(\vA) = r \leq R$, and ${\bf{col} (\vU)}$ as the column space of $\vU$. Fix $\eps, \eta \in (0, 1)$ such that $N \lesssim 1/\eps^{r}$, and vector $\mathbf{c} \in \Real^N$.  Draw a degree-$d$ KFJLT $\vPhi \in \mathbb{C}^{m \times N}$, with $m$ satisfying
\begin{equation}
\label{m term}
m \geq C'\,\Big[\eps^{-2}\, \log^{2d-1}\left( \frac{d}{\eps^r \eta} \right)\, \log^4\left(\frac{\log^d(\frac{d}{\eps^r \eta})}{\eps}\right)\, \log N\Big].
\end{equation}
Then, with respect to the draw of $\vPhi$,
\begin{equation*}
\vP\Big(\|\vPhi\vy\|_2^2  = (1\pm \eps)\|\vy\|_2^2, ~\forall\vy \in {\bf{col} (\vU)} - \mathbf{c} \Big) \geq 1-\eta-2^{-\Omega(\log N)}.
\end{equation*}
Here, $C' > 0$ in \cref{m term} is a universal constant.
\end{corollary}

\cref{subspace embedding} is obtained from \cref{main} via a standard covering argument; specifically, by applying the KFJLT $\vPhi$ on a net of ${\bf{col} (\vU)} - \mathbf{c}$ with cardinality $O(1/\eps^r)$ \cite{LGM96} with the distortion factor $\eps/2$ \cite{BDDW07}.

\begin{proof}[Proof of \cref{sample size}]

Suppose $\vPhi$ is an $m \times N$ KFJLT with parameters defined as in \cref{sample size}. 

Apply \cref{m term} from \cref{subspace embedding} to each of $\mathbf{c}_1  = {\bf 0}$ and $\mathbf{c}_2 = {\bf b}$, using parameters $\eps'=\sqrt{\eps/r}$ and $\eta'=\eta/2$:
\begin{enumerate}
{\setlength\itemindent{12pt} 
\item 
$\vPhi$ is a $(1\pm \sqrt{\eps/r})$ embedding on ${\bf{col} (\vU)}$;}
{\setlength\itemindent{12pt} 
 \item 
 $\vPhi$ is a $(1\pm \sqrt{\eps/r})$ embedding on the affine space $\bf{col} (\vU)-\vbbb$.}
\end{enumerate}
\cref{subspace embedding}, followed by the union bound, gives that the above conditions hold simultaneously with probability at least $1-2\,(\eta/2+ 2^{-\Omega(\log N)})\geq 1-\eta - 2^{-\Omega(\log N)},$ provided that 
\begin{align*}
m & = O\left(\eps^{-1}\,r\, \log^{2d-1}\left( (\frac{r}{\eps})^{\frac{r}{2}}\,\frac{d}{\eta} \right) \, \log^4\left( (\frac{r}{\eps})^{\frac{1}{2}}\,\log^d((\frac{r}{\eps})^\frac{r}{2}\,\frac{d}{\eta}) \right)
\, \log N\right).
\end{align*}
With the assumption on $d/\eta \leq (r/\eps)^{r/2}$, we can get to that
\begin{align*}
m & = O\left(\eps^{-1}\,r \,\log^{2d-1}\left( (\frac{r}{\eps})^r \right) \, \log^4\left((\frac{r}{\eps})^{\frac{1}{2}}\,\log^d((\frac{r}{\eps})^r) \right)
\, \log N\right)\\
& = O\left(\eps^{-1}\,r^{2d} \,\log^{2d-1}( \frac{r}{\eps}) \, \log^4\left( (\frac{r}{\eps})^{\frac{1}{2}}\,r^d \,\log^d(\frac{r}{\eps}) \right)\, \log N\right)\\
& = O\left(\eps^{-1}\,r^{2d} \,\log^{2d-1}( \frac{r}{\eps}) \, \log^4\left(\frac{r^{d+\frac{1}{2}}}{\eps^{\frac{1}{2}}}\, \log^{d}( \frac{r}{\eps})\right)\, \log N\right)\\
& = O\left(\eps^{-1}\,r^{2d} \,\log^{2d-1}( \frac{r}{\eps}) \,(d+\frac{1}{2})^4\, \log^4\left(\frac{r}{\eps}\,\log( \frac{r}{\eps})\right)\, \log N\right)\\
& = O\left(\eps^{-1}\,r^{2d} \,\log^{2d-1}( \frac{r}{\eps}) \, \log^4\left(\frac{r}{\eps}\,\log( \frac{r}{\eps})\right)\, \log N\right).
\end{align*}

Following the proof of Theorem $2.16$ in \cite{W14}, under such conditions, the solution $\hat{\vxx}$ to the sketched least squares  \cref{sketched solution} achieves
\begin{equation}
\label{success sketch ls}
\| \vA\hat{\vxx}- \vbbb \|_2  = (1\pm O(\eps)) \min_{\vxx\in \Real^R}\| \vA\vxx - \vbbb \|_2.
\end{equation}

\end{proof}

\section{Bounding the embedding dimension in \cref{main}}
\subsection{Background review}
The proof draws on a result established in \cite{KW11} showing that matrices which can stably embed \emph{sparse} vectors  -- or have a certain \emph{restricted isometry property} (RIP) \cite{CRT06, CT06, FR13} -- result in Johnson-Lindenstrauss embeddings if their column signs are randomly permuted.  First let us recall the definition of the RIP.
\begin{definition}
\label{rip}
A matrix $\vPsi \in \mathbb{C}^{m\times N}$ is said to have the RIP of order $T$ and level $\delta\in (0,1)$ ($(T, \delta)$-RIP) if 
\begin{equation}
\|\vPsi \vxx\|_{2}^{2} = (1\pm \delta)\|\vxx\|_{2}^{2}\quad\text{for all $T$-sparse $\vxx\in \mathbb{R}^{N}$}.
\end{equation}
A vector is $T$-sparse if it has at most $T$ nonzero entries.
\end{definition}
The main result of \cref{prop:KW} says that randomizing the columns signs of a $(T, \delta)$-RIP matrix results in a random JL embedding on a fixed set of $p = O(\exp(T))$ points with multiplicative distortion $4\delta$.
\begin{theorem}[Theorem 3.1 from \cite{KW11}]
\label{prop:KW}
Fix $\eta > 0$ and $\eps\in (0,1)$ and consider a finite set $\sE \subset \mathbb{R}^N$ of cardinality $| \sE | = p$.  
Set $s \geq 20 \log(4p/\eta)$ and suppose that $\vPsi \in \mathbb{R}^{m \times N}$ satisfies the $(2s,\delta)$-RIP and $\delta \leq \eps/4$.  Let $\xi \in \mathbb{R}^N$ be a Rademacher vector. Then 
\begin{equation*}
\vP\Big(\| \vPsi \vD_{\xi} \vxx\| _{2}^{2} = (1\pm\varepsilon)\| \vxx\| _{2}^{2}, ~~\forall \vxx \in \sE\Big) \geq 1-\eta.
\end{equation*}
\end{theorem}

The randomly-subsampled discrete Fourier matrix $\sqrt{N/m} \,\vS \vF \in \mathbb{C}^{m \times N}$ is known to satisfy the restricted isometry property with nearly-optimally small embedding dimension $m$ \cite{RV08, CGV13, KMR14, B14, HR16}.  More generally, any randomly-subsampled unitary matrix $\vU$ whose entries are uniformly bounded $\| \vU \|_{\infty} \leq O(1/\sqrt{N})$ also satisfies the RIP with such $m$.  The sharpest known bound on $m$ for such constructions is from \cite{HR16}, stated below.

\begin{proposition}[Theorem 1.1 from \cite{HR16}]
\label{RIPsharp}
For sufficiently large $N$ and $T$, a unitary matrix $\vU \in\mathbb{C}^{N\times N}$ satisfying $\|\vU\|_\infty \leq O(1/\sqrt{N})$, and a sufficiently small $\delta > 0$, the following holds. For some 
\begin{equation}
\label{rip m}
m(T, \delta) = O\left(\delta^{-2} \, T \, \log^2(T/\delta) \, \log^2(1/\delta) \, \log N\right),
\end{equation}
let $\vPsi = \sqrt{N/m} \, \vS\vU\in \mathbb{C}^{m \times N}$. Then, with probability $1 - 2^{-\Omega(\log N\, \log(T/\delta))}$, the matrix $\vPsi$ satisfies $(T, \delta)$-RIP.
\end{proposition}

Now we relate these known results to the KFJLT \cref{eq:kfjlt_example}.  Since each of the component DFT matrices $\vF_{k}$ is unitary and the Kronecker product preserves orthogonality, it follows $\bigotimes_{k=d}^{1}\vF_{k}$ is a unitary matrix of size $N,$ denoted by $\vU$.  Moreover, since each $\| \vF_{k} \|_{\infty} = 1/\sqrt{n_k}$, it follows that $\| \vU \|_{\infty} = 1/\sqrt{N}$ recalling $N$ is the product of each mode size $n_k$.  Finally, by the distributive property of the Kronecker product, the KFJLT in \cref{eq:kfjlt_example} is equivalent to:
\begin{equation}
\label{transform}
\sqrt{\frac{N}{m_{\text{kron}}}}\,\vS\left(\bigotimes_{k=d}^{1}\vF_{k}\right)\left(\bigotimes_{k=d}^{1}\vD_{\xi_{k}}\right)
\displaystyle= \sqrt{\frac{N}{m_{\text{kron}}}}\,\vS\vU\left(\bigotimes_{k=d}^{1}\vD_{\xi_{k}}\right)
= \sqrt{\frac{N}{m_{\text{kron}}}}\,\vS\vU\vD_\zeta,
\end{equation}
where the random vector $\zeta= \bigotimes_{k=d}^{1}\xi_{k}\in\Real^N$ is the Kronecker product of the Rademacher vectors $\xi_{k}$ for $k \in [d]$.

Thus, the KFJLT is constructed from an optimal RIP matrix (as in \cref{RIPsharp}), and with column signs randomized according to the structured random vector $\zeta= \bigotimes_{k=d}^{1}\xi_{k}$.     
In the special case $d=1$, the KFJLT reduces to the standard FJLT, and the embedding result of \cref{main} is obtained from \cref{prop:KW} directly.   In the case $d\geq 2$, it remains to analyze the effect of applying the Kronecker Rademacher vector $\zeta$, as opposed to an i.i.d. Rademacher vector, to the RIP matrix.

\subsection{Concentration inequality}
We here introduce a more general version of \cref{prop:KW}, which works for any degree-$d$ construction consisting of a RIP matrix with randomized column sign from a Kronecker product of $d$ independent Rademacher vectors.

\begin{theorem}
\label{analog}
Fix $d \geq 2$, and $\delta \in (0, 1/2s^{d-1})$, $\eta \in (0,1)$. Let $n_1, n_2, \dots, n_d$ be positive integers and $N = \prod_{k=1}^d n_k$. Consider a finite set $\sE\subset \mathbb{R}^N$ of cardinality $|\sE|=p \geq N$. Let $\xi_{1} \in \Real^{n_1}, \dots, \xi_{d} \in \Real^{n_d}$ be independent Rademacher vectors and $\zeta = \bigotimes_{k=d}^{1}\xi_{k} \in \Real^N$. Set $\eps = 2\delta\, s^{d-1}\in (0,1)$ and the condition on $s$:
\begin{equation}
\label{condition}
\displaystyle \max_{k\in [d]} n_k \geq s  \geq 128\log\left((d+2)\, N^{2-\frac{2}{d}}\, \frac{2p}{\eta}\right) \geq 66.
\end{equation}
If $\vPsi\in\mathbb{R}^{m\times N}$ is a $(2s, \delta)$-RIP matrix, then
\begin{equation}
\label{success}
\vP\Big(\| \vPsi \vD_{\zeta} \vxx\| _{2}^{2} = (1\pm\varepsilon)\| \vxx\| _{2}^{2}, ~~\forall \vxx \in \sE\Big) \geq 1-\eta.
\end{equation}
\end{theorem}

\begin{remark}
\label{complexification}
\cref{analog} is stated for real-valued embeddings, though the KFJLTs are in the complex field. The result extends to complex matrices straightforwardly via a standard complexification strategy described below. Suppose a partial Fourier matrix $\vPsi = \vPsi_1 + \vi\,\vPsi_2 \in \mathbb{C}^{m \times N}$ with $\vPsi_1, \vPsi_2 \in \Real^{m \times N}$, we map the embedding to a $2m$ tall matrix $\tilde{\vPsi}$ with $\vPsi_1$ on top and $\vPsi_2$ bottom. The new real-valued matrix satisfies the same RIP if $\vPsi$ has this property by equivalence of their operator norms if $\vxx, \zeta\in\Real^N$ are real-valued.
\begin{equation*}
\|\tilde{\vPsi}\vD_\zeta\vxx\|_2^2 = \|\left[
\begin{array}{l}
\vPsi_1\\
\vPsi_2
\end{array}
\right]\vD_\zeta\vxx\|_2^2 = \|\vPsi_1\vD_\zeta\vxx\|_2^2+\|\vPsi_2\vD_\zeta\vxx\|_2^2=\|\vPsi\vD_\zeta\vxx\|_2^2.
\end{equation*}
Rescaling the result for real-valued embeddings by a factor $1/2$, we obtain the bound of $m_{\text{kron}}$.
\end{remark}

To prove \cref{analog}, we use the following probability bound result.
\begin{proposition}
\label{prob bound}
Fix $d \geq 2$, and an integer $s$ satisfying $66 \leq s \leq \max_{k\in [d]} n_k$, $\delta \in (0, 1/2s^{d-1}).$ Let $n_1, n_2, \dots, n_d$ be positive integers and $N = \prod_{k=1}^d n_k$. Consider a fixed vector $\vxx \in \Real^N$, and suppose $\xi_{1} \in \Real^{n_1}, \dots, \xi_{d} \in \Real^{n_d}$ are independent Rademacher vectors and $\zeta = \bigotimes_{k=d}^{1}\xi_{k} \in \Real^N$. If $\vPsi\in\mathbb{R}^{m\times N}$ is a $(2s, \delta)$-RIP matrix, then
\begin{equation}
\label{simplify}
\vP\Big( \Big|\|\vPsi\vD_\zeta \vxx\| _{2}^{2} - \| \vxx\| _{2}^{2}\Big| > 2\delta\, s^{d-1}  \, \| \vxx\| _{2}^{2} \Big)
\leq 2(d+2)\, N^{2-\frac{2}{d}}\,\exp(-\frac{1}{128}\, s).
\end{equation}
\end{proposition}

The proof of \cref{prob bound} is shown in \cref{proof of prob bound}. We now prove \cref{analog}.
\begin{proof}[Proof of \cref{analog}]
For fixed $\eps = 2\delta\, s^{d-1} \in (0,1), \eta\in (0,1)$, we focus on the event $E$:  $\vPsi \vD_\zeta$ fails to embed at least one of the $p$ vectors in $\sE$ within $\eps$ distortion ratio, which is the complement event of \cref{success}. We apply \cref{prob bound} with $p$ times the bound in \cref{simplify} to get that
\begin{equation*}
\vP(E)\leq 2p\,(d+2)\, N^{2-\frac{2}{d}}\,\exp(-\frac{1}{128}\, s).
\end{equation*} 

Given the conditions \cref{condition} in \cref{analog}, we have
\begin{equation*}
\vP(E) \leq \displaystyle2p\, (d+2)\, N^{2-\frac{2}{d}}\, \exp(-\log((d+2)\, N^{2-\frac{2}{d}}\, \frac{2p}{\eta})) = \eta,
\end{equation*}
hence 
\begin{equation*}
\vP\Big(\| \vPsi \vD_{\zeta} \vxx\| _{2}^{2} = (1\pm\varepsilon)\| \vxx\| _{2}^{2}, ~~\forall \vxx \in \sE\Big) = 1-\vP(E) \geq 1-\eta.
\end{equation*}
\end{proof}

\subsection{Proof of \cref{main}}
\begin{proof}
The case $d=1$ corresponds to standard FJLT and is already known \cite{KW11}; here we focus on $d\geq 2$. 

We now derive \cref{main} from \cref{analog,RIPsharp}. Recall the degree-$d$ KFJLT $\vPhi = \sqrt{N/m_{\text{kron}}} \,\vS\vU_N\vD_\zeta \in \mathbb{C}^{m_\text{kron} \times N}$ shown in \cref{transform}. Following from \cref{RIPsharp}, let $\vPsi = \sqrt{N/m_{\text{kron}}}\,\vS\vU_N$, we use the complexification technique in \cref{complexification} and analyze the construction $\tilde{\vPsi}\vD_\zeta$ instead. 

For fixed $\eps, \eta \in (0,1)$, consider a vector set $\sE \subset \Real^N$ of cardinality $p$, suppose $s$ satisfies the condition \cref{condition}  in \cref{analog} and $\tilde{\vPsi}\in\Real^{2m_\text{kron} \times N}$ has the $(2s, \delta)$-RIP property.  

Without loss of generality, suppose $p \geq N$, because we can always embed a set of vectors of cardinality less than $N$ from $\Real^N$ in $\Real^p$.  Hence from \cref{condition}
\begin{equation}
\label{rough cond1}
s = \displaystyle O\Big(\log(\frac{(d+2)\, p^{3-\frac{2}{d}}}{\eta})\Big) \leq  O\Big(\log\big((\frac{dp}{\eta})^{3-\frac{2}{d}}\big)\Big) = O\Big(\log(\frac{dp}{\eta})\Big),
\end{equation}
then from $\eps = 2\delta \, s^{d-1}$, we have
\begin{equation}
\label{rough cond2}
\delta = \displaystyle \frac{\eps}{s^{d-1}}
= o\Big(\frac{\eps}{\log^{d-1}(\frac{dp}{\eta})}\Big).
\end{equation}

We now apply \cref{RIPsharp}: for $N, s$ sufficiently large and $\delta$ sufficiently small, we obtain an upper bound on $m_{\text{kron}}$ from \cref{rip m}:
\begin{align}
\nonumber
&m_{\text{kron}} = \frac{m(2s,\delta)}{2}  = O\left(\delta^{-2}\, s \, \log^2(\frac{1}{\delta})\, \log^2(\frac{s}{\delta})\,\log N\right)\\
\nonumber
&\leq  O\left(\eps^{-2}\,\log^{2d-2}(\frac{dp}{\eta})\,\log(\frac{dp}{\eta})\,\log^2(\frac{\log^{d-1}(\frac{dp}{\eta})}{\eps})\, \log^2(\frac{\log^d(\frac{p}{\eta})}{\eps})\, \log N\right)\\
\label{m bound}
&\leq  O\left(\eps^{-2}\,\log^{2d-1}(\frac{dp}{\eta})\, \log^4(\frac{\log^d(\frac{dp}{\eta})}{\eps})\, \log N\right).
\end{align}

If we choose $m$ to be at least the order of \cref{m bound}, then 
$\vP\Big(\|\vPhi\vxx\|_2^2 = (1\pm \eps)\|\vxx\|_2^2\Big) \geq 1-\eta.$ Noting finally that the small probability of failure $2^{-\Omega(\log N\, \log(T/\delta))}$ in \cref{RIPsharp} is bounded above by $2^{-\Omega(\log{N})}$, the proof of the main result \cref{main} is complete.
\end{proof}

\section{Bounding the probability in \cref{prob bound}}
\label{proof of prob bound}

\subsection{Proof ingredients}
We recall basic corollaries of the restricted isometry property, whose proofs can be found in \cite{R10}.

\begin{enumerate}
\item
\begin{lemma}\cite{R10}
\label{RIP-norm 1}
Let $\vxx\in\mathbb{R}^{n}$ be a vector and suppose that $\vPsi\in\mathbb{R}^{m\times n}$ has the $(2s,\delta)$-RIP. Then, for an index  subset $I \subset [n]$ of size $|I|\leq 2s$,
\begin{equation}
\|\vPsi(:, I)\vxx(I)\|_{2}^{2} = (1\pm \delta)\|\vxx(I)\|_2^2.
\end{equation}
\end{lemma}

\item
\begin{lemma}\cite{R10}
\label{RIP-norm 2}
Let $\vxx\in\mathbb{R}^{n}$ be a vector and suppose that $\vPsi\in\mathbb{R}^{m\times n}$ has the $(2s,\delta)$-RIP. Then, for any pair of disjoint index subsets $I, J\subset [n]$ of size $|I|, |J|\leq s$,
\begin{equation}
\langle \vPsi(:, I)\vxx(I), \vPsi(:, J)\vxx(J)\rangle =\pm \delta\|\vxx(I)\|_{2}\,\|\vxx(J)\|_{2}.
\end{equation}
\end{lemma}
\end{enumerate}

Then let us recall standard concentration inequalities in both linear and quadratic forms, particularly for Rademacher vectors:
\begin{lemma}[Hoeffding's inequality]
\label{Hoeffding}
Let $\vxx\in\mathbb{R}^{n}$ be a vector and $\xi\in\mathbb{R}^{n}$ be a Rademacher vector. Then, for any $t>0$,
\begin{equation}
\vP(|\xi^{\top}\vxx|>t)\leq 2\exp{(-\frac{t^{2}}{2\|\vxx\|_{2}^{2}})}.
\end{equation}
\end{lemma}
This version of Hoeffding's inequality is derived directly from Theorem 2 of \cite{H63}.

\begin{lemma}[Hanson-Wright inequality]\cite{HW71}
\label{Hanson-Wright}
Let $\vx\in\mathbb{R}^{n\times n}$ have zero diagonal entries, and $\xi\in\mathbb{R}^{n}$ be a Rademacher vector. Then, for any $t>0$,
\begin{equation}
\vP(|\xi^{\top} \vx \xi|>t)\leq 2\exp{(-\frac{1}{64}\min{(\frac{t^{2}}{\| \vx \|_{F}^{2}}, \frac{\frac{96}{65}t}{\|\vx\|})})}.
\end{equation}
\end{lemma}
This Hanson-Wright bound with explicit constants is derived from the proof of Theorem 17 in \cite{BLM03}.

We will use the following corollary of Hanson-Wright.
\begin{corollary}
\label{induction}
Suppose we have a random matrix $\vx \in \mathbb{R}^{n \times n}$,  positive vectors $\vy_1, \vy_2 \in \mathbb{R}^{n},$ and $\vt> 0$, $\vp > 0$ (assume $\vp < 1/n^2$) such that, for each pair $(i,j) \in [n]$,
\begin{equation*}
    \vP\left(|x(i, j)| > \vt \, \vy_1(i)\, \vy_2(j)\right) \leq \vp.
\end{equation*}
Then for a Rademacher vector $\xi \in \mathbb{R}^{n}$, and $t > 0$ such that $\vt \leq t/66$, we have
\begin{equation}
    \vP \left(|\xi^\top \vx \xi| > t\, \|\vy_1\|_2\, \|\vy_2 \|_2\right) \leq n^2\, \vp + 2\exp(-\frac{1}{44}\, \frac{t}{\vt}),
\end{equation}
where the probability is with respect to both $\vx$ and $\xi$.
\end{corollary}

\begin{proof}
With probability at least $1 - n^2 \,\vp$ with respect to the draw of $\vx$, $\{|x(i,j)| \leq \vt\, \vy(i) \, \vy(j)\}$ for all $i,j \in [n]$.

By the law of total probability,
\begin{align*}
\begin{array}{l}
\vP\left(|\xi^\top \vx\xi| > t\, \| \vy_1\| _2 \, \| \vy_2\| _2\right)  \\\\
\leq n^2\, \vp +\vP\left(|\xi^\top \vx\xi| > t\, \| \vy_1\| _2 \, \| \vy_2\| _2 ~ \Big| ~ \{|x(i,j)| \leq \vt\, \vy_1(i) \, \vy_2(j)\} \text{ for all } i,j \in [n]   \right) \nonumber \\\\
\leq n^2\, \vp +\vP\left(|\xi^\top \vx\xi| > t\, \| \vy_1\| _2 \, \| \vy_2\| _2~ \Big| ~ {E}   \right),
\end{array}
\end{align*}
where $E$ is the event
$$
{E} = \{ |\text{Tr}(\vx)| \leq \vt\,\| \vy_1\| _2 \, \| \vy_2\| _2, \quad  \|\tilde{\vx}\| \leq \vt\,\| \vy_1\| _2 \, \| \vy_2\| _2, \quad \|\tilde{\vx}\|_{F} \leq \vt\,\| \vy_1\| _2 \, \| \vy_2\| _2  \},
$$
and $\text{Tr}(\vx)$ is the trace of $\vx$, and $\tilde{\vx}$ is formed from $\vx$ by setting the diagonal entries to zero.  

Then, applying \cref{Hanson-Wright},
\begin{align*}
\begin{array}{l}
\vP\left(\Big|\xi^\top \vx\xi\Big| > t\,\| \vy_1\| _2 \, \| \vy_2\| _2~ \Big| ~ {E} \right)\\\\
\leq \vP\left(|\text{Tr}(\vx)| > \vt\,\| \vy_1\| _2 \, \| \vy_2\| _2~ \Big| ~ {E} \right)  +\vP\left(\Big|\xi^{\top}\tilde{\vx}\xi\Big| > (t - \vt)\,\| \vy_1\| _2 \, \| \vy_2\| _2~ \Big| ~{E} \right)\\\\
= \vP\left(\Big|\xi^{\top}\tilde{\vx}\xi\Big| > (t - \vt)\,\| \vy_1\| _2 \, \| \vy_2\| _2~ \Big| ~{E} \right)\\\\
\leq 2\displaystyle\exp\left( -\min(\frac{1}{64}\,\frac{(t - \vt)^{2}}{\vt^{2}}, \frac{3}{130}\,\frac{t - \vt}{\vt}) \right) = 2\exp(-\frac{3}{130}\,\frac{t - \vt}{\vt})\\\\
\leq 2\displaystyle\exp(-\frac{1}{44}\,\frac{t}{\vt})~~~(\text{if}~ \frac{t}{\vt} \geq 66).
\end{array}
\end{align*}

Therefore we obtain a upper bound: 
$$
\vP\left(\Big|\xi^\top \vx\xi\Big| > t\,\| \vy_1\| _2 \, \| \vy_2\| _2\right) \leq n^{2}\, \vp+2\exp(-\frac{1}{44}\,\frac{t}{\vt}).
$$
\end{proof}

The following proposition concerns norm bounds for a class of (possibly asymmetric) RIP-based quadratic forms.

\begin{proposition}
\label{modified}
Fix integers $s, n, m$ such that $s \leq n$ and let $r = \lceil n/s \rceil \geq 1$. Consider vectors $\vxx, \vy \in \Real^n$. Let $\vPsi = (\vPsi_L, \vPsi_R) \in \Real^{m \times 2n}$, where $\vPsi_L, \vPsi_R \in \Real^{m \times n}$ respectively denote the first and the second sets of $n$ columns, have the $(2s, \delta)$-RIP. After sorting the indices by the magnitude of the entries in $\vxx$, let $I_1$ denote the first $s$ sorted indices, $I_2$ denote the second $s$ (possibly less than $s$) sorted indices, and up to $I_r$, where $|I_1| = \dots = |I_{r-1}| = s$ and $|I_r| = n-(r-1)\, s$. The corresponding index notations for $\vy$ are the sets $J_1, \dots, J_r$. We write $i \sim j$ if the two indices are associated in the same block location respectively of $\vxx$ and $\vy$, i.e. $i \in I_p, j \in J_p$, $p \in [r]$. Consider the matrix $\vC_{\vxx, \vy} \in \Real^{n\times n}$ with entries:
\begin{equation*}
\vC_{\vxx, \vy}(i, j) = 
\left\{
\begin{array}{lc}
x(i)\vPsi_L(:,i)^\top \vPsi_R(:,j)y(j), & i \not\sim j,~~ i \in I_1^c, j \in J_1^c,\\
0, & \text{else}.
\end{array}
\right.
\end{equation*}
And for $\vbbb\in \{-1, 1\}^{s}$, $\vdd \in \{-1, 1\}^n$,
\begin{align*}
\vv_{\vxx, \vy}  &= \vD_{\vxx(I_1^c)}\vPsi_L(:, I_1^c)^\top \vPsi_R(:, J_1)\vD_{\vy(J_1)}\vbbb \in \Real^{n-s},\\
w_{\vxx, \vy} & = \sum_{p=1}^{r}\vdd(I_p)^\top \vD_{\vxx(I_p)}\vPsi_L(:, I_p)^\top \vPsi_R(:, J_p)\vD_{\vy(J_p)}\vdd(J_p)\in \Real.
\end{align*}
Then
\begin{align*}
\|\vC_{\vxx, \vy}\| \leq \frac{\delta}{s}\, \|\vxx\|_2\, \|\vy\|_2, 
&~~~~\|\vC_{\vxx, \vy}\|_F \leq \frac{\delta}{\sqrt{s}}\, \|\vxx\|_2\, \|\vy\|_2,\\ \|\vv_{\vxx, \vy}\|_2 \leq \frac{\delta}{\sqrt{s}}\, \|\vxx\|_2\, \|\vy\|_2,
&~~~~|w_{\vxx, \vy}| \leq \delta\, \|\vxx\|_2\, \|\vy\|_2.
\end{align*}
\end{proposition}
The detailed proof of \cref{modified} can be found in \cref{modified proof}.

\begin{remark}
\label{extension}
Proposition 5.4 from \cite{KW11} can be viewed as a special case of \cref{modified} in which the quadratic form is symmetric.
\end{remark}

\subsection{Notation}
\label{notation}
Let us define the notation for this section.
\begin{enumerate}
\item Suppose $\xi_1 \in \Real^{n_1}, \dots, \xi_d \in \Real^{n_d}$ are independent Rademacher vectors, let $N = \prod_{k=1}^d n_k,$ and define 
$$
\vD_{\zeta} = \bigotimes^{1}_{k=d}\vD_{\xi_k} \in \Real^{N \times N}.
$$
We consider the diagonal block form of $\vD_{\zeta}:$
\begin{equation}
\label{D block}
\vD_\zeta \displaystyle= \vD_{\xi_d}\otimes \left(\bigotimes^{1}_{k=d-1}\vD_{\xi_k}\right).
\end{equation}

\item $\vPsi \in \Real^{m \times N}$ is a deterministic matrix with corresponding column block decomposition 
\begin{equation}
\label{block}
\vPsi = \left(\vPsi_1, \vPsi_2, \dots, \underbrace{\vPsi_i}_{m \times \prod_{k=1}^{d-1} n_k}, \dots, \vPsi_{n_d}\right) \in \Real^{m \times \prod_{k=1}^d n_k}.
\end{equation}

\item $\vPhi \in \Real^{m \times N}$ is the random matrix $\vPhi = \vPsi \vD_{\zeta}$.
We also consider each $\vPhi_i \in \Real^{m \times \prod_{k=1}^{d-1} n_k}$ defined by the corresponding block forms of $\vPsi, \vD_\zeta: $
$$\vPhi_i = \displaystyle \vPsi_i\, \left(\bigotimes^{1}_{k=d-1}\vD_{\xi_k}\right), \quad \text{for}~~i\in[n_d].$$
Importantly, $\vPhi_i$ with randomized signs is a column block of $\vPhi$ by \eqref{D block} \eqref{block}:
\begin{equation}
\label{Phi block}
\vPhi = \left(\xi_{d}(1)\, \vPhi_1, \dots, \xi_{d}(i)\, \vPhi_i, \dots, \xi_{d}(n_d)\, \vPhi_{n_d}\right) \in \Real^{m \times \prod_{k=1}^d n_k}.
\end{equation}

\item For fixed $\vxx \in \Real^N$, we consider the corresponding vector block decomposition 
\begin{equation}
\vxx = \left(\vxx_1, \vxx_2, \dots, \underbrace{\vxx_i}_{\in\Real^{\prod_{k=1}^{d-1} n_k}}, \dots, \vxx_{n_d}\right)^\top \in \Real^{\prod_{k=1}^d n_k}.
\end{equation}
\end{enumerate}

Note that the distortion 
\begin{equation}
\label{dist}
\|\vPhi \vxx\|_2^2 - \|\vxx\|_2^2 = \|\vPsi \vD_{\zeta}\vxx\|_2^2-\|\vxx\|_2^2
\end{equation}
can be equivalently expressed as the quadratic form
$
\zeta^{\top}\vM\zeta,
$
where $\vM = \vD_{\vxx}(\vPsi^\top\vPsi- {\mathbf{I}_N})\vD_{\vxx}\in \mathbb{R}^{N\times N}$.

Based on the block decomposition in \cref{Phi block}, the distortion \cref{dist} can also be expressed as the quadratic form 
\begin{equation}
\label{qua2}
\xi_{d}^\top \vM_d \xi_{d},
\end{equation}
where $\vM_d \in \mathbb{R}^{n_d\times n_d}$ has entries $\{m_d(i,j)\}_{i,j=1}^{n_d}$ given by
\begin{align}
\label{m_d}
m_d(i, j) & = \left\{
\begin{array}{lc}
   \|\vPhi_i\vxx_i\|_2^2-\|\vxx_i\|_2^2, & i = j \\
   \langle\vPhi_i\vxx_i, \vPhi_j\vxx_j\rangle, & i \neq j
\end{array}
\right..
\end{align}

\subsection{Proof of \cref{prob bound}}
We prove \cref{prob bound} for a general $d \geq 2$ via induction on the degree $d$.  Note that \cref{prob bound} with degree $d$ is simply the following proposition restricted to $i=j$ of degree $d+1$.

\begin{proposition}
\label{generalized prob bound}
 Fix $d\geq 2$ and consider vectors $\{\vxx_i\}_{i \in [n_d]}\in\Real^{\prod_{k=1}^{d-1} n_k}$.  Let $\vPsi \in \Real^{m \times N}$ and $\vPhi \in \Real^{m \times N}$ be as defined above.  Fix an integer $s$ satisfying $66 \leq s \leq \max_{k\in [d]} n_k$, and $\delta \in (0,1/(2 s^{d-2}))$. Suppose that $\vPsi$ has the $(2s, \delta)$-RIP. Then, for each pair $(i, j) \in [n_d]\times [n_d]$,
\begin{align}
\label{md concentration}
\begin{array}{lc}
\vP\Big(|m_d(i, j)|> 2\delta\, s^{d-2}  \, \|\vxx_i\|_2\,\|\vxx_j\|_2\Big) \leq 
(6\displaystyle\prod_{k=2}^{d-1}  n_k^2+2\sum_{k=2}^{d-1}\prod_{\ell = k+1}^{d-1} n_\ell^2)\,\exp(-\frac{1}{128}\, s),
\end{array}
\end{align}
where we define $\prod_{k=d}^{d-1} n_k^2 =1$ and $\prod_{\ell=3}^1 n_\ell^2 =0$.
\end{proposition}

\begin{proof}[Proof of \cref{generalized prob bound}]

\begin{enumerate}[leftmargin=*, labelindent =0.2cm, listparindent=1.5em]
\item For the base case $d=2$, the matrix $\vM_2 \in \Real^{n_2 \times n_2}$ as defined in \cref{m_d} has entries
\begin{equation}
\label{m2}
m_2(i, j)  = \xi_{1}^\top\vA_{i, j}\xi_{1},    
\end{equation} 
where the matrix $\vA_{i, j} \in \Real^{n_1\times n_1}$ is given by
\begin{equation*}
\left\{
\begin{array}{lc}
\vD_{\vxx_i}(\vPsi_i^\top\vPsi_i -\Id_{n_1})\vD_{\vxx_i},    & i =j, \\
\vD_{\vxx_i}\vPsi_i^\top\vPsi_j\vD_{\vxx_j},    & i \neq j.
\end{array}
\right.
\end{equation*}
The proof of the proposition in the case  $i=j$ follows from Proposition 5.4 in \cite{KW11} used to analyze the standard FJLT.  However,  the case $i \neq j$ requires an extension of \cref{prop:KW} in \cite{KW11} to quadratic forms that are not necessarily symmetric, which constitutes our Proposition \ref{modified}.  

After sorting the indices by the magnitude of the entries in $\vxx_i$, let $I_1$ denote the first $s$ sorted indices, $I_2$ denote the second set of $s$ sorted indices, and up to $I_r$, where $|I_1| = \dots = |I_{r-1}| = s$ and $|I_r| = n_1-(r-1)\, s$. The corresponding index notations for $\vxx_j$ are the sets $J_1, \dots, J_r$. We write $i_1 \sim j_1$ if the two indices are associated in the same block location respectively of $\vxx$ and $\vy$, i.e. $i_1 \in I_p, j_1 \in J_p$, $p \in [r]$.

\begin{remark}
A significant difference between our Proposition \ref{modified}  and Proposition 5.4 in \cite{KW11} is that the entries in the vector $\vxx$ of length $N$ cannot be ordered by magnitude freely in our case. We are able to only sort the entries within each vector block of length $n_1.$ The reason is the corresponding diagonal entries in the matrix $\vD_{\zeta}$ are no longer independent, where $\zeta = \bigotimes_{k=d}^1 \xi_{k}.$
\end{remark}

Consider the matrix $\vC_{i, j} \in \Real^{n_1\times n_1}$ with entries:
\begin{equation*}
\vC_{i, j}(i_1, j_1) = 
\left\{
\begin{array}{lc}
m_{i,j}(i_1, j_1), & i_1 \not\sim j_1, i_1 \in I_1^c, j_1 \in J_1^c,\\
0, & \text{else}.
\end{array}
\right.
\end{equation*}

And, 
\begin{align*}
\vv_{i, j}  &= \vA_{i,j}(I_1^c, J_1)\xi_{1}(J_1) \in \Real^{n_1-s},\\\\
w_{i, j} & = \sum_{p=1}^{r}\xi_{1}(I_p)^\top \vA_{i, j}(I_p, J_p)\xi_{1}(J_p) \in \Real.
\end{align*}

Each of the blocks $\{\vPsi_i\}_{i \in [n_2]}$ has the $(2s, \delta)$-RIP because $\vPsi$ has the $(2s, \delta)$-RIP and $s \leq n_1$ by assumption.  Therefore, we apply the result from \cref{modified} to derive the norm bounds. In particular, the symmetric case when $i=j$ refers to the special case when $\vPsi_L = \vPsi_R$ and $\vxx = \vy$ in \cref{modified}.
\begin{align*}
\|\vC_{i, j}\| \leq \frac{\delta}{s}\, \|\vxx_i\|_2\, \|\vxx_j\|_2, 
&~~~~\|\vC_{i, j}\|_F \leq \frac{\delta}{\sqrt{s}}\, \|\vxx_i\|_2\, \|\vxx_j\|_2,\\
\|\vv_{i, j}\|_2 \leq \frac{\delta}{\sqrt{s}}\, \|\vxx_i\|_2\, \|\vxx_j\|_2,
&~~~~|w_{i,j}| \leq \delta\,\|\vxx_i\|_2\, \|\vxx_j\|_2.
\end{align*}

Moreover, 
\begin{align*}
m_2(i, j) & = \xi_{1}^\top \vA_{i, j}\xi_{1} = \sum_{p, q =1}^{r}\xi_{1}(I_p)^\top \vA_{i, j}(I_p, J_q)\xi_{1}(J_q)\\
&= \xi_{1}^\top \vC_{i, j}\xi_{1}+ \xi_{1}(I_1^c)^\top\vv_{i,j}+\xi_{1}(J_1^c)^\top\vv_{j,i}+w_{i,j},
\end{align*}
since 
\begin{align*}
\xi_{1}^\top \vC_{i, j}\xi_{1} &=  \sum_{p, q =2 \atop p \neq q}^{r}\xi_{1}(I_p)^\top \vA_{i, j}(I_p, J_q)\xi_{1}(J_q),\\
\xi_{1}(I_1^c)^\top\vv_{i,j}& = \sum_{p=2}^{r}\xi_{1}(I_p)^\top \vA_{i, j}(I_p, J_1)\xi_{1}(J_1),\\
\vv_{j,i}^\top\xi_{1}(J_1^c) & = \sum_{q=2}^{r}\xi_{1}(I_1)^\top \vA_{i, j}(I_1, J_q)\xi_{1}(J_q) = \xi_{1}(J_1^c)^\top\vv_{j,i}.
\end{align*}
By the standard concentration inequalities \cref{Hoeffding} and \cref{Hanson-Wright},
\begin{align*}
\vP\left(|\xi_{1}(I_1^c)^\top\vv_{i,j}| > \frac{1}{8}\, \delta\,\|\vxx_i\|_{2} \,\|\vxx_j\|_2\right)& \leq 2\exp\left(-\frac{1}{2}\,\frac{\delta^2}{64}\, \frac{s}{\delta^2}\right)\\
&= 2\exp(-\frac{1}{128}\, s) ,\\\\
\vP\left(|\xi_{1}(J_1^c)^\top\vv_{j, i}| > \frac{1}{8}\, \delta\,\|\vxx_i\|_{2} \,\|\vxx_j\|_2\right)& \leq 2\exp(-\frac{1}{128}\, s),\\\\
\vP\left(|\xi_{1}^\top \vC_{i, j}\xi_{1}| > \frac{3}{4}\, \delta\,\|\vxx_i\|_{2} \,\|\vxx_j\|_2\right)& \leq 2\exp\left(-\min(\frac{1}{64}\, \frac{9\delta^2}{16}\, \frac{s}{\delta^2}, \frac{3}{130}\, \frac{3\delta}{4}\, \frac{s}{\delta})\right)\\
&= 2\exp{\left(-\min(\frac{9}{1024}\, s, \frac{9}{520}\, s)\right)}\\
&<  2\exp(-\frac{1}{128}\, s).
\end{align*}

Since 
\begin{equation*}
\left\{
\begin{array}{lc}
\displaystyle|\xi_{1}(I_1^c)^\top\vv_{i,j}| \leq \frac{1}{8}\, \delta\,\|\vxx_i\|_{2} \,\|\vxx_j\|_2,\\\\
\displaystyle|\xi_{1}(J_1^c)^\top\vv_{j, i}| \leq \frac{1}{8}\, \delta\,\|\vxx_i\|_{2} \,\|\vxx_j\|_2,\\\\
\displaystyle|\xi_{1}^\top \vC_{i, j}\xi_1| \leq\frac{3}{4} \, \delta\,\|\vxx_i\|_{2} \,\|\vxx_j\|_2,\\\\
|w_{i, j}|\leq \delta\,\|\vxx_i\|_{2} \,\|\vxx_j\|_2,
\end{array}
\right.
\end{equation*}
imply $|m_2(i, j)| \leq 2\delta\,\|\vxx_i\|_{2} \,\|\vxx_j\|_2$, by the law of total probability, we obtain that for each fixed pair $(i, j) \in [n_2]\times [n_2]$,
\begin{align*}
\vP\left(\Big|m_2(i, j)\Big| > 2\delta\, \|\vxx_i\|_{2} \,\|\vxx_j\|_2\right) &< 2\exp(-\frac{1}{128}\, s)+2\exp(-\frac{1}{128}\, s)+2\exp(-\frac{1}{128}\, s)\\
& \leq 6\exp(-\frac{1}{128}\, s).
\end{align*}

\item Suppose now that \cref{generalized prob bound} is true up to degree $d$ ($d \geq 2$). We aim to show that the statement must then hold also for degree $d+1$. 

In the degree-$(d+1)$ case, for fixed $i\in [n_{d+1}]$, consider one further step of the block decomposition:
\begin{align*}
\vxx_i &= \left(\vxx_{i, 1}, \vxx_{i, 2}, \dots, \vxx_{i, i'}, \dots, \vxx_{i, n_d}\right)\in \Real^{\prod_{k=1}^{d}n_k},\\
\vPsi_i &= \left(\vPsi_{i, 1}, \vPsi_{i, 2}, \dots, \vPsi_{i, i'}, \dots, \vPsi_{i, n_d}\right)\in\Real^{m \times \prod_{k=1}^{d}n_k}.
\end{align*}

Recall the form \cref{qua2} of $m_{d+1}$. 
Note that when the degree jumps to $d+1$, there are now $n^2_{d+1}$ matrices $\{(\vM_d)_{i,j}\}_{i,j\in[n_{d+1}]}$ in the similar forms of $\vM_d$ \eqref{m_d}.
As a consequence, $\vM_{d+1}$ expands from simply a number $m_{d+1}$ to becoming a $n_{d+1}$ by $n_{d+1}$ square matrix, where each quadratic form squeezes $(\vM_d)_{i,j}$ into a number and thus form an entry of $\vM_{d+1}.$ More specifically, 
\begin{equation*}
\vM_{d+1} = 
\underbrace{\left[\begin{array}{lllll}
\xi_d^\top(\vM_d)_{1,1}\xi_d&&&&\\
&\ddots&&&\\
&&\underbrace{\xi_d^\top\underbrace{(\vM_d)_{i,j}}_{\in \Real^{n_d\times n_d}}\xi_d}_{\in \Real}&&\\
&&&\ddots&\\
&&&&\xi_d^\top(\vM_d)_{n_{d+1}, n_{d+1}}\xi_d
\end{array}\right]}_{n_{d+1}\times n_{d+1}}.
\end{equation*}
For each pair of $(i, j) \in [n_{d+1}]\times [n_{d+1}]$,
\begin{equation*}
m_{d+1}(i, j) = \xi_{d}^\top (\vM_d)_{i, j}\xi_{d},
\end{equation*}
where the matrix $(\vM_d)_{i, j}\in\Real^{n_d\times n_d}$ has entries:
\begin{align*}
(m_d)_{i,j}(i', j')& = 
\left\{
\begin{array}{lc}
 \|\vPhi_{i, i'}\vxx_{i, i'}\|_2^2 - \|\vxx_{i, i'}\|_2^2,    & \text{if}~~ i = j ~~\text{and}~~ i' = j',\\
 \langle \vPhi_{i, i'}\vxx_{i, i'}, \vPhi_{j, j'}\vxx_{j, j'}\rangle,  & \text{else},
\end{array}
\right.
\end{align*}
have the form $m_d$ from \cref{m_d} since each $\vPhi_{i, i'}$ is a $(2s, \delta)$-RIP $\vPsi_{i, i'}$ with randomized column signs from $\bigotimes_{k=d-1}^1 \xi_{k}$ and the index sets of $\vPhi_{i, i'}, \vPhi_{j, j'}$ are disjoint except when both $i = j$ and $i' = j'$.

By the concentration bound for degree $d$ which we assume is true,  
\begin{align}
\label{generalized prob bound term}
\begin{array}{l}
\nonumber
\vP\Big( |(m_d)_{i, j}(i', j')|>2\delta\, s^{d-2} \, \|\vxx_{i, i'}\|_2\,\|\vxx_{j, j'}\|_2 \Big)\\\\
\leq 
(6\displaystyle\prod_{k=2}^{d-1} n_k^2+2\sum_{k=2}^{d-1}\prod_{\ell = k+1}^{d-1} n_\ell^2)\,\exp(-\frac{1}{128}\, s).
\end{array}
\end{align}

Thus with the assumption $s = (2\delta \, s^{d-1})/(2\delta \, s^{d-2})\geq 66$, we now apply \cref{induction}, to get that
\begin{equation*}
\begin{array}{l}
\vP\Big(|m_{d+1}(i, j)|> 2\delta\, s^{d-1} \, \|\vxx_i\|_2\,\|\vxx_j\|_2\Big) \\\\
\leq \displaystyle n_d^2 \, [(6\prod_{k= 2}^{d-1} n_k^2+2\sum_{k=2}^{d-1}\prod_{\ell = k+1}^{d-1} n_\ell^2)\,\exp(-\displaystyle\frac{1}{128}\, s)]+ 2\exp(-\frac{1}{44}\,\frac{2\delta\, s^{d-1}}{2\delta\, s^{d-2}})\\\\
\displaystyle\leq (6\prod_{k= 2}^d n_k^2+2\sum_{k=2}^{d-1}\prod_{\ell = k+1}^{d} n_\ell^2)\,\exp(-\displaystyle\frac{1}{128}\, s)+2\exp(-\frac{1}{44}\, s)\\\\
\displaystyle\leq (6\prod_{k= 2}^d n_k^2+2\sum_{k=2}^{d}\prod_{\ell = k+1}^{d} n_\ell^2)\,\exp(-\displaystyle\frac{1}{128}\, s).
\end{array}
\end{equation*}
Note that we simply replace $\exp(-s/44)$ with $\exp(-s/128)$ in the last step derivation as the latter term is of bigger value and works for an upper bound to make a more organized result. 
\end{enumerate}

We finally obtain the result for $d+1$ and complete the proof of \cref{generalized prob bound}.
\end{proof}

We now show the proof of \cref{prob bound}.

\begin{proof}[Proof of \cref{prob bound}]
Let $n_1^* \geq n_2^*\geq \dots \geq n_d^*$ be the decreasing arrangement of $\{n_k\}_{k\in [d]}$. As the setting in \cref{prob bound} with degree $d$ is the case of $m_{d+1}$  restricted to $i=j$ shown in \cref{generalized prob bound}, we obtain
\begin{equation}
\label{prob bound term}
\vP\Big( \Big|\|\vPsi\vD_\zeta \vxx\| _{2}^{2} - \| \vxx\| _{2}^{2}\Big| > 2\delta\, s^{d-1}  \, \| \vxx\| _{2}^{2} \Big)
\leq\displaystyle\underbrace{[6\prod_{k= 2}^d (n_k^*)^2+2\sum_{k=2}^d\prod_{\ell = k+1}^{d} (n_\ell^*)^2]}_{C_d}\displaystyle\,\exp(-\frac{1}{128}\, s).
\end{equation}
We can further reduce the bound in \cref{prob bound term} to
$$
2(d+2)\, N^{2-\frac{2}{d}}\,\exp(-\frac{1}{128} \, s),
$$
due to $C_d \leq (6+2(d-1)) \,(N/n^*_1)^2$ and $n_1^* \geq N^{1/d}$ when $d \geq 2$. The proof is complete.
\end{proof}

\begin{remark}
The RIP assumption in \cref{generalized prob bound} is only used in the base case $d=2$, and not in the induction step.  This indicates that our proof is not optimal; it remains an interesting open problem to improve the resulting bound from $m \sim \eps^{-2} \log(1/\eta)^{2d-1}$ to $m \sim \eps^{-2} \log(1/\eta)^{d+1}$.
\end{remark}

\section{Numerical experiments and further discussions}
\label{numerics}
In this section, we run numerical experiments to study the empirical embedding performance of  KFJLT. It is of value to discuss and compare the performances of KFJLT with varying degree $d$, including the standard FJLT corresponding to $d=1$, in order to evaluate the trade-off between distortion power and computational speed-up. Here we define the embedding distortion ratio:
\begin{equation*}
\left\vert\frac{\|\vPhi\vxx\|_2^2-\|\vxx\|_2^2}{\|\vxx\|_2^2}\right\vert.
\end{equation*}

\subsection{FJLT vs Kronecker FJLT}

FJLT and KFJLT differ in the mixing operation. We show the numerical result Figure 2 comparing the embedding performance of standard FJLT, degree-$2$ and degree-$3$ KFJLTs on a set of randomly constructed Kronecker vectors. The numerical observation suggests that KFJLTs take slightly more rows to achieve the same quality of embeddings and lose some stability compared to standard FJLT, which is consistent with the theory.

\begin{figure}[!ht]
\label{ratio}
\centering 
\subfloat[uniform $\lbrack0,1\rbrack$]{\includegraphics[width=.48\textwidth]{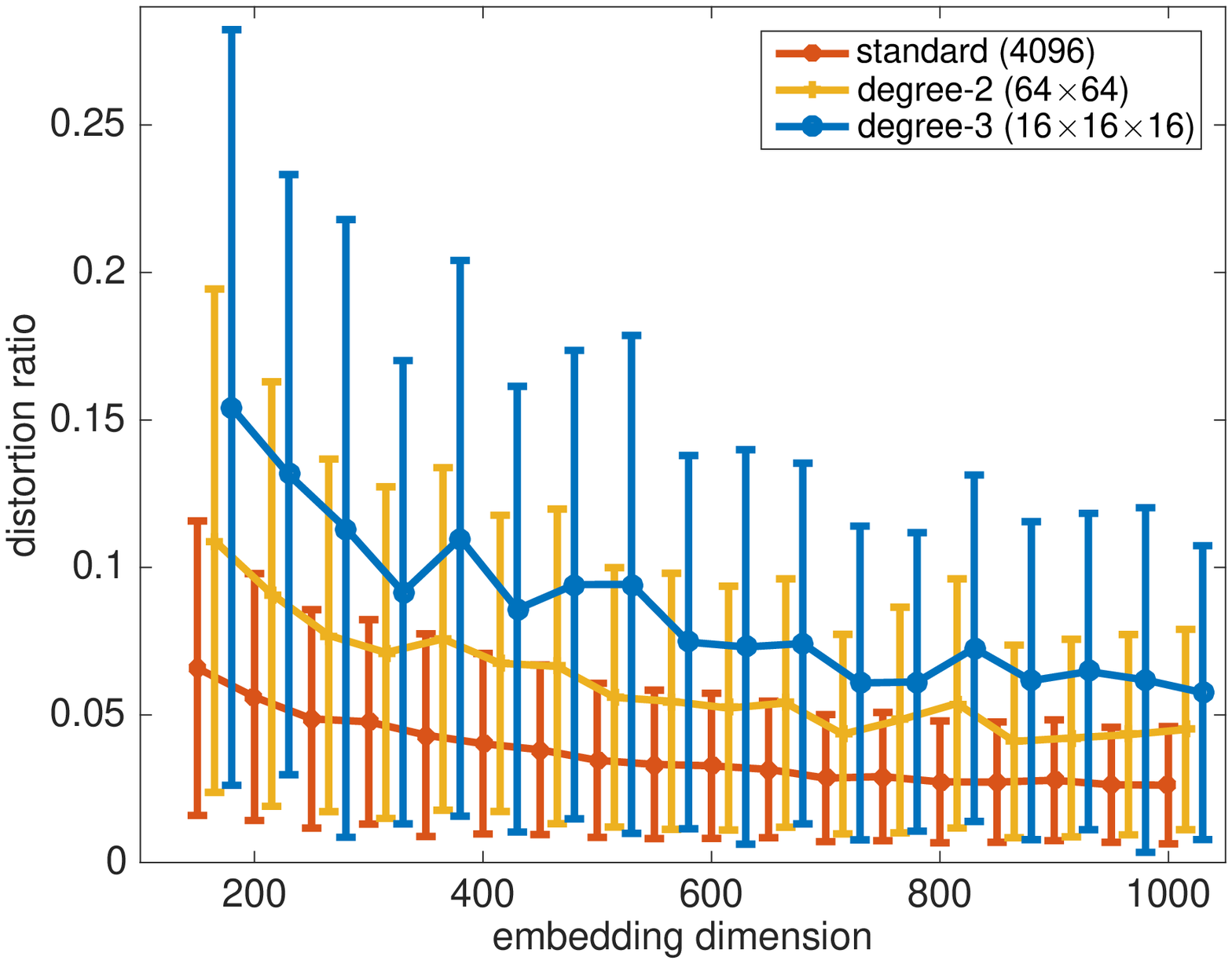}}~~
\subfloat[normally distributed]{\includegraphics[width=.48\textwidth]{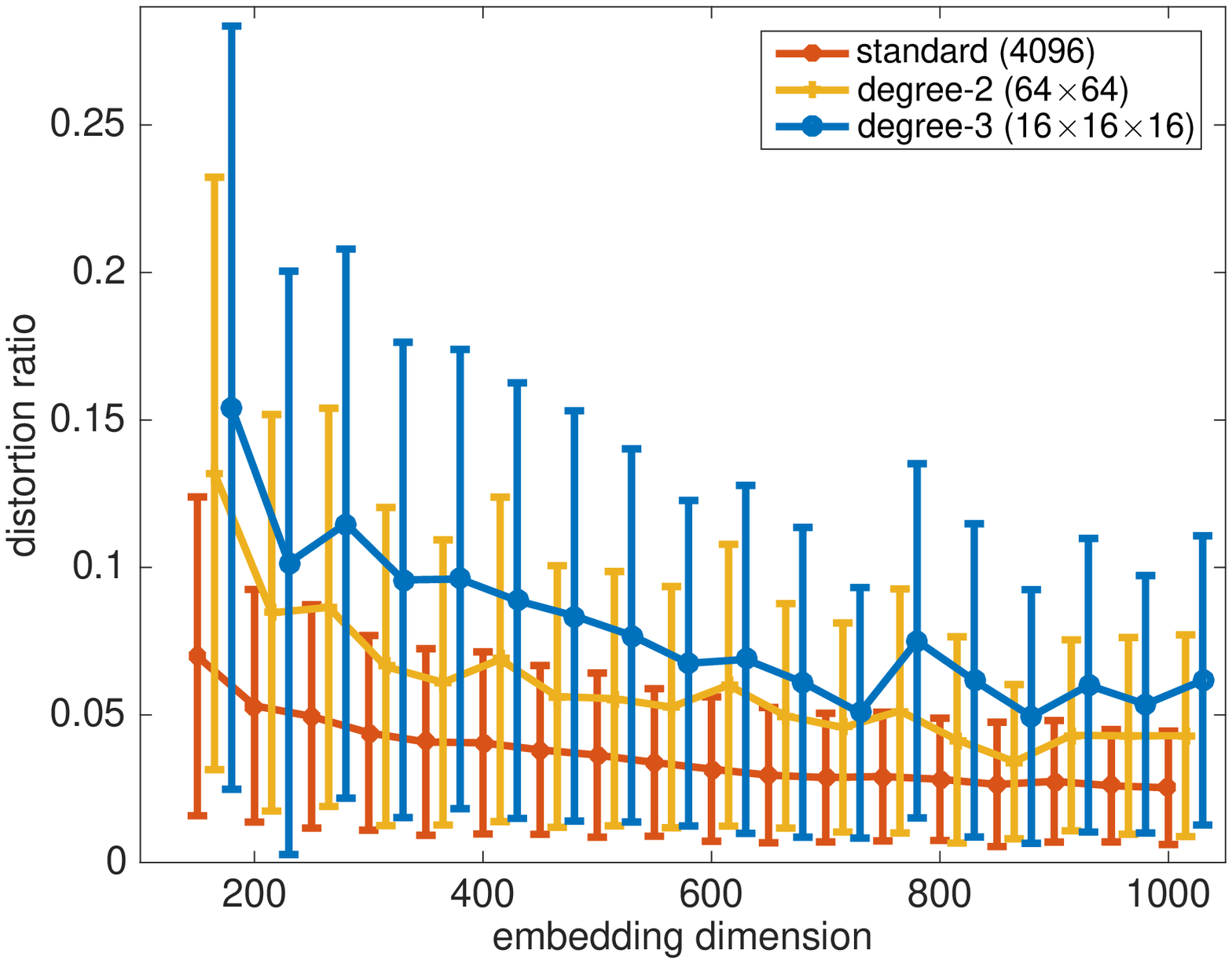}}~~
\caption{Comparing the embedding performance between the standard FJLTs and the KFJLTs of degree $2$ and $3$. Each dot with error bar represents the average distortion ratio and the standard deviation based on $1000$ trials for a given embedding dimension. In each trial, we generate the subsampled standard or Kroneckerized FFTs and the i.i.d. or Kroneckerized random sign-flipping operations for the three constructions respectively and test them on the same vector. The vectors to be embedded are $(\Real^4)^{\bigotimes 6}$ Kronecker vectors, hence simultaneously degree-$2$ and degree-$3$ Kronecker vectors, and they consist of respectively uniform $[0,1]$ (on the left) and normally distributed elements (on the right) in each component vector. The x-axis of the degree-$2$ and degree-$3$ plots are shifted by $+15, +30$ respectively.}
\end{figure}

\subsection{Kronecker-structured vs general vectors}
\label{unstructured}
It is clearly of interest to study the general case of KFJLT embedding arbitrary Euclidean vectors since it is needed for the theoretical analysis of CPRAND-MIX algorithm, though KFJLT is designed to accelerate dimension-reduction for tall Kronecker-structured vectors. One might also wonder if the main embedding results can be improved if we just restrict to Kronecker vectors, but the experiments Figure 3 suggests that, the Kronecker-structured vectors result in worst-case embedding compared to general random vectors. 

\begin{figure}[!ht]
\label{unstructured plot}
\centering 
\subfloat[Degree-$2 ~(125\times 125)$ KFJLT]{\includegraphics[width=.48\textwidth]{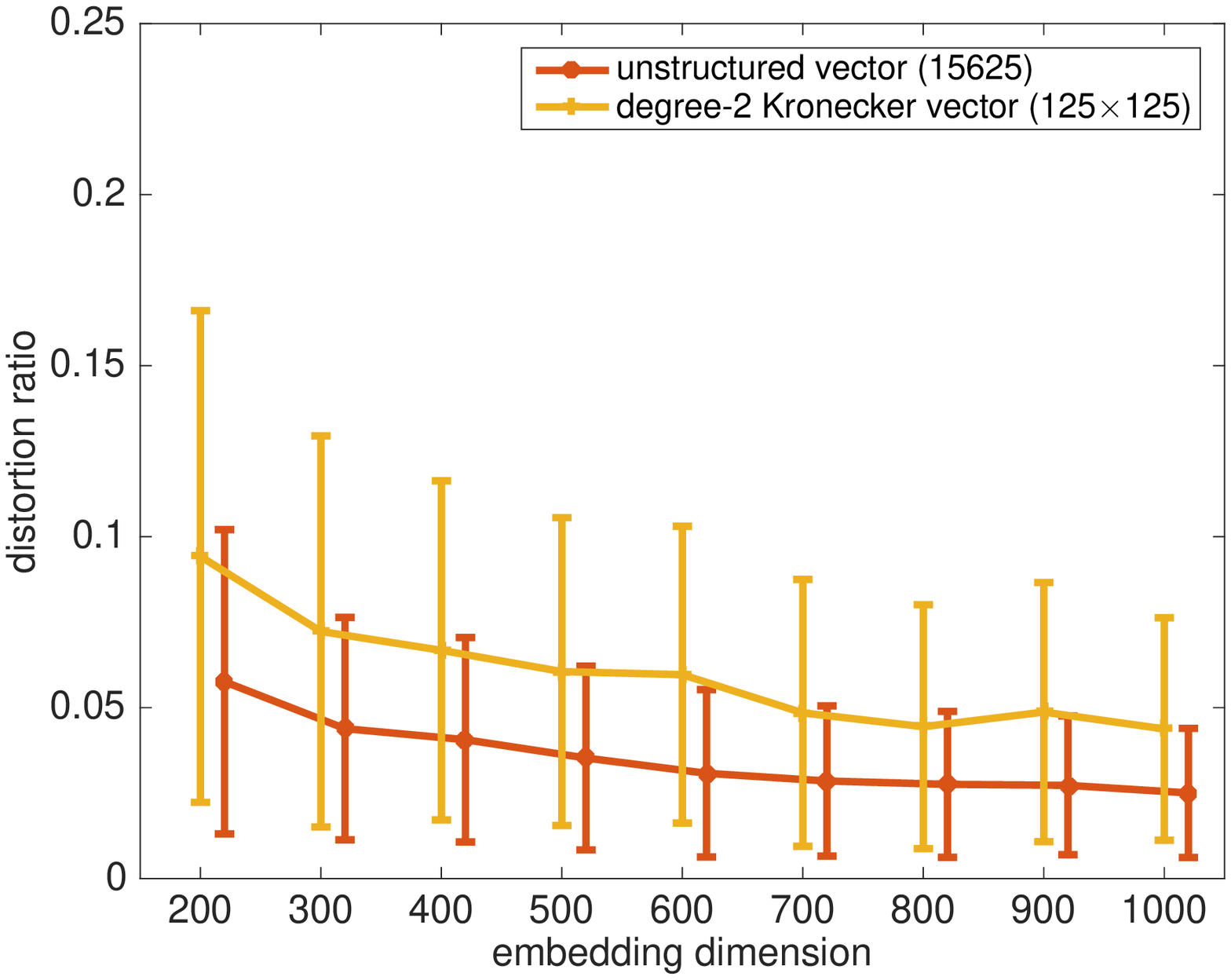}}~~
\subfloat[Degree-$3 ~(25\times 25 \times 25)$ KFJLT]{\includegraphics[width=.48\textwidth]{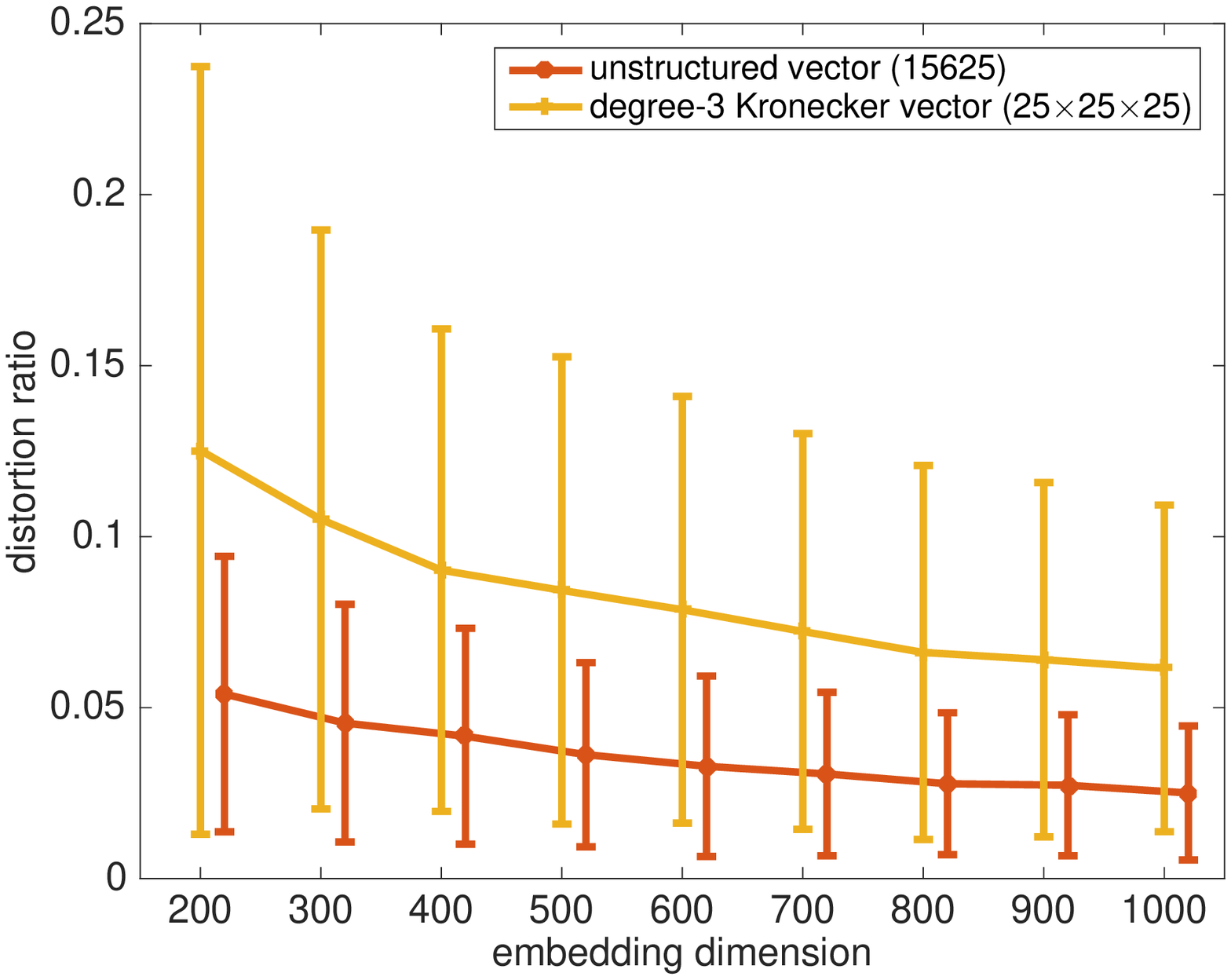}}
\caption{Comparing the embedding performance of the KFJLT embedding on general and Kronecker Euclidean vectors. Each dot with error bar represents the average distortion ratio and the standard deviations based on $1000$ trials for a given embedding dimension. In each trial, we generate a general vector as well as a Kronecker vector and embed each of them using the same KFJLT. Each vector consists of normally distributed elements, either in full or each component vector. The x-axis of the unstructured vector plot is shifted by $+20$.}
\end{figure}

To understand how the Kronecker structure contributes to the gap, we go back to the technical proof. One possible reason could be that, from the result shown in \cref{prob bound}, the probability bound in \cref{simplify} is determined by 
$$(d+2)\, N^{2-\frac{2}{d}}\,\vp/3,
$$
where we denote $\vp$ as the probability bound for $|m_2(i, j)|$ concentrating in the scale $2\delta$ for $i, j \in [n_2]$, which serves as the base bound in the overall derivation (in our result, $\vp = 6\exp(-s/128)$). Given a certain tolerance $\eps = 2\delta\, s^{d-1} \in (0,1)$, it is more unlikely to control the distortion with a bigger base bound. 

When $i \neq j$, by a general version of Hanson-Wright inequality \cite{RV08}, 
\begin{equation*}
\vP(\vert m_2(i, j) - \mathbb{E}\left(m_2(i, j)\right) \vert > t) \leq 2\exp~[-c\,\min(\frac{t^2}{\|\vM_{i, j}\|_F^2}, \frac{t}{\|\vM_{i, j}\|})],
\end{equation*}
$\vp$ increases if $m_2(i, j)$ tends to concentrate around a greater expectation. Moreover,
\begin{equation*}
\mathbb{E}\left(m_2(i, j)\right) = \sum_{i_1=1}^{n_1} x_i(i_1)\, x_j(i_1)\, \vPsi_i(i_1)^\top \vPsi_j(i_1),
\end{equation*}
the correlation between entries in blocks $\vxx_i$ and $\vxx_j$ can make a difference in the estimation of $\mathbb{E}\left(m_2(i, j)\right)$.  Following the \emph{rearrangement inequality}, the expectation tends to reach its highest value among all the choices of pairwise arrangements when $x_j(i_1)$ is in the same position as $x_i(i_1)$ after reordering according to their decreasing arrangements. Vectors with Kronecker structure happen to be in this particular situation, thus achieving larger distortion in general, compared to general vectors. 

One other possible explanation for the discrepancy showed in Figure 3 is that the Kronecker product of Gaussians is no longer Gaussian. The lack of good properties from the normal distribution might contribute to the gap.

\subsection{Sampling strategy in KFJLT}
In constructing the KFJLT, it might seem less natural to first construct the Kronecker product $\bigotimes_{k=d}^{1} \vF_{k}\vD_{\xi_{k}}$ and then subsample rows uniformly (sample after Kronecker), as we propose, compared to first uniformly subsampling each $\vF_{k}$ and then taking the Kronecker product of the resulting subsampled matrices (sample before Kronecker).  On the one hand, the sampling operation does not affect the computational savings for KFJLT, hence there is no major difference in the computational cost between two sampling methods.  However, uniformly subsampling in the final step as we do does lead to a better JL embedding.

Indeed, consider instead sampling components $\sqrt{n_k/m_k}\,\vS_k\in \mathbb{R}^{m_{k}\times n_{k}}$ for $k\in [d]$, and forming the alternative embedding
\begin{equation}
\label{alternate}
\vPhi \vxx=\left(\bigotimes\limits_{k=d}^{1}\sqrt{\frac{n_k}{m_k}}\,\vS_{k}\vF_{k}\vD_{\xi_{k}}\right)\left(\bigotimes\limits_{k=d}^{1}\vxx_k\right)=\bigotimes\limits_{k=d}^{1}\underbrace{(\sqrt{\frac{n_k}{m_k}}\,\vS_{k}\vF_{k}\vD_{\xi_{k}})}_{\text{standard FJLT:}~\vPhi_k}\vxx_k.
\end{equation}

We have the distortion estimation:
\begin{equation*}
\|\vPhi \vxx\|_2^2 - \|\vxx\|_2^2 = O\left(\max\limits_{k\in[d]}\Big\vert\|\vPhi_k \vxx_k\|_2^2 - \|\vxx_k\|_2^2\Big\vert\right).
\end{equation*}

To achieve a $(1\pm\varepsilon)$ approximation, each $m_k$ must be of the scale $\varepsilon^{-2}$ based on \cref{eq:fjlt_example}. Hence the total embedding dimension $m = \prod_{k=1}^{d}m_k$ must be at least of the order $\varepsilon^{-2d}$, which is significantly worse than the scaling we obtain with uniform sampling, $\varepsilon^{-2}$.
 
We corroborate this calculation empirically below in Figure 4, comparing the distortions resulting from our KFJLT with those resulting from a Kronecker-factored sampling strategy as in \cref{alternate}.

\begin{remark}
Note that there are other practical tensor-structured sampling strategies besides for the one introduced above. For example, \cite{CLNW19} and \cite{SGTU18} consider the row-wise Kronecker sampling. Such row-wise design not only reduces the embedding cost with the help of the Kronecker factorization but also keeps the distortion from growing exponentially with the degree $d$. We refer the readers to the first discussion in Section $1.3$ for details. Also, a recent new paper \cite{LK20} proposed an efficient leverage-score based sampling strategy for sparse tensor-structured matrices.
\end{remark}

\begin{figure}[!ht]
\centering 
\label{sampling}
\subfloat[Degree-$2 ~(125\times 125)$]{\includegraphics[width=.48\textwidth]{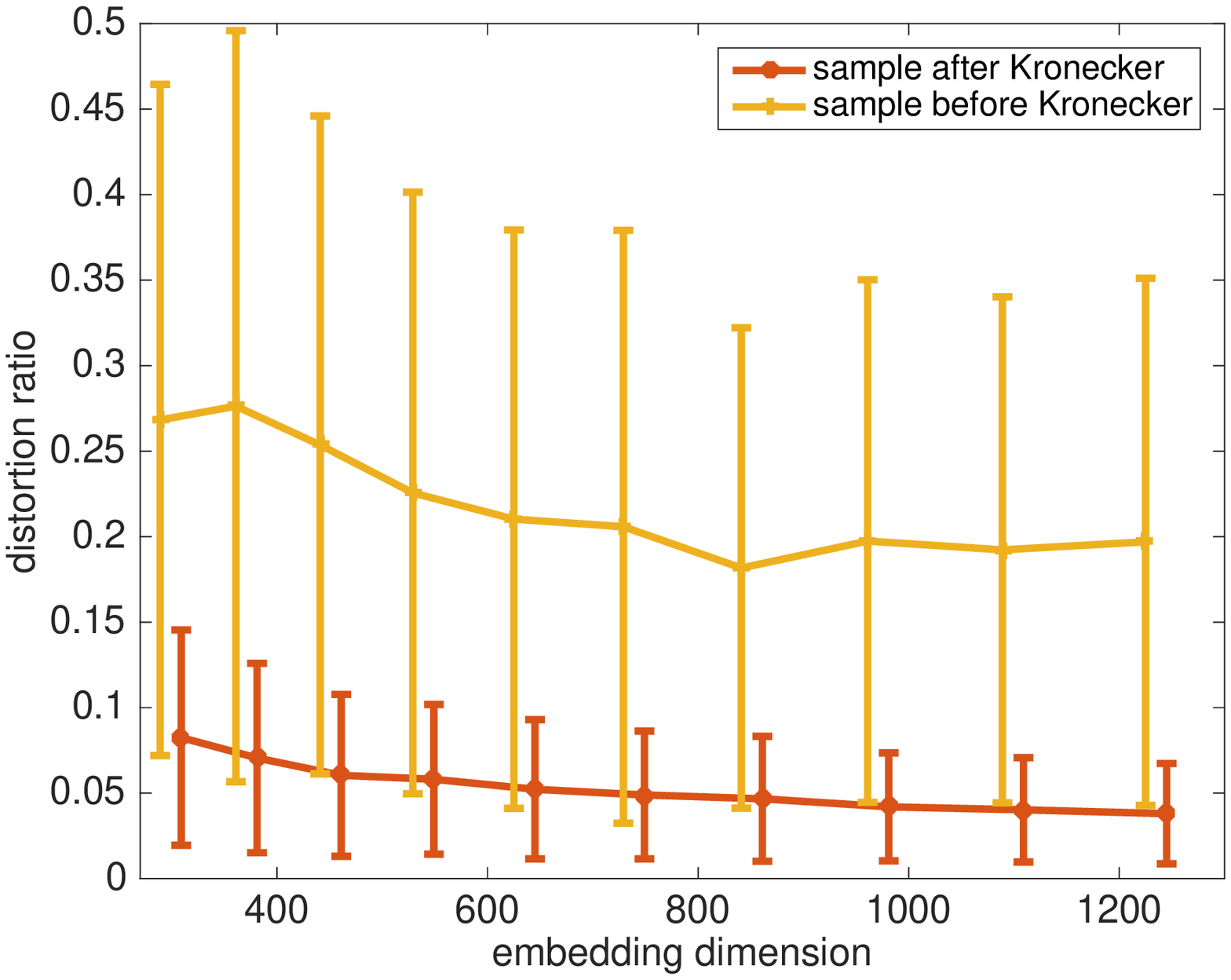}}~~
\subfloat[Degree-$3 ~(25\times 25 \times 25)$]{\includegraphics[width=.48\textwidth]{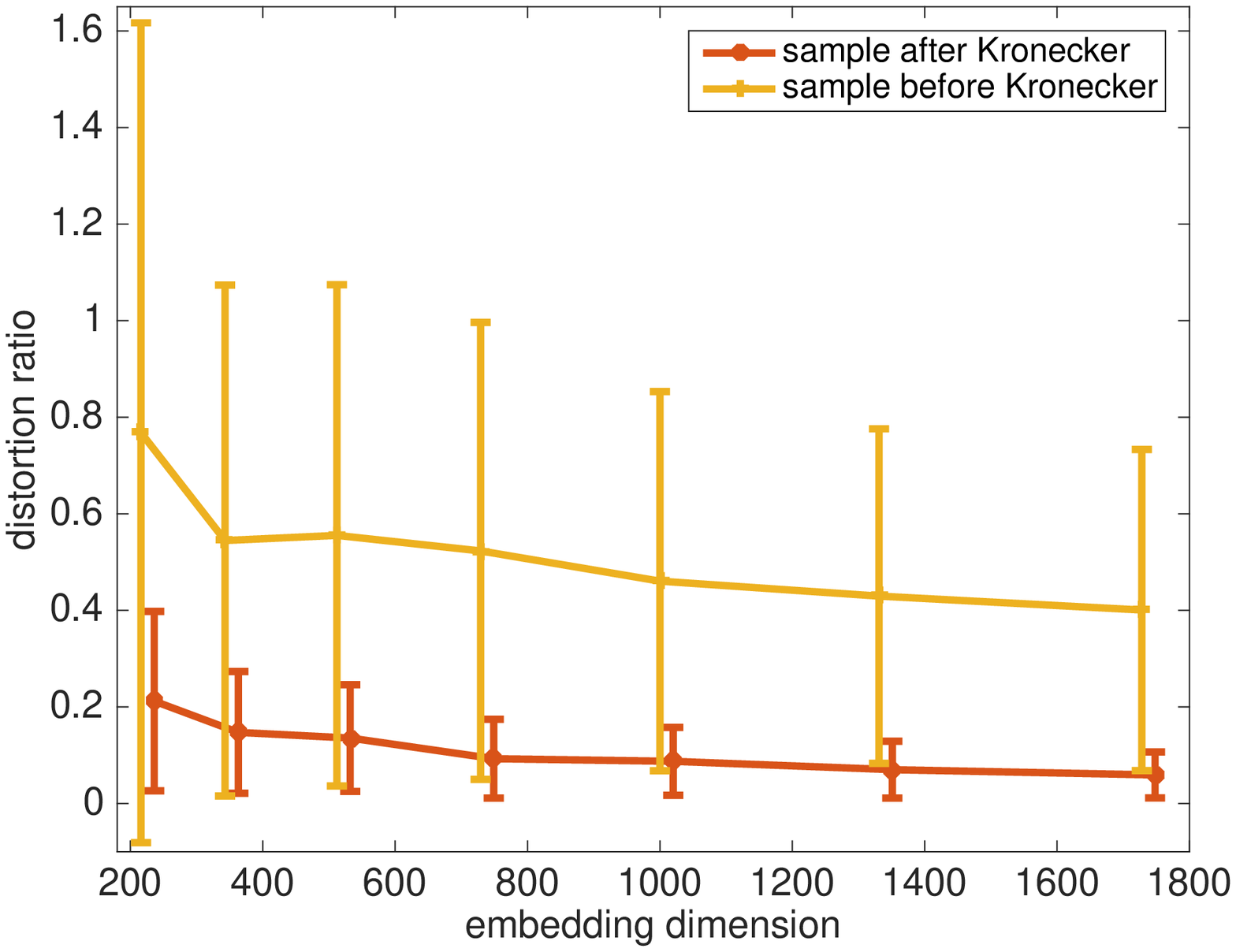}}
\caption{We compare the embedding performance in first constructing the Kronecker product and then sampling the rows uniformly,  to the performance in first uniformly sampling each constituent Fourier transform and then taking the Kronecker product of the results. Each dot with error bar represents the average distortion ratio and standard deviation based on $1000$ trials for a given embedding dimension. In each trial, we generate the same sign-flipped FFT for each Kronecker component but different sampling instructions on the same embedding dimension for two constructions. We test them on the same vector. The embedded objects are respectively degree-$2$ and degree-$3$ Kronecker vectors in the two plots. They consist of normally distributed elements in each component vector. The x-axis of the uniform sampling plot is shifted by $+20$.}
\label{sampling comparison}
\end{figure}

\section*{Acknowledgments}
The authors would like to thank Bosu Choi, Joe Neeman, Eric Price, Per-Gunnar Martinsson, George David Torres, David Woodruff and Gordan \u Zitokovi\'c for valuable help and discussions. Part of this work was done while Rachel Ward and Ruhui Jin were visiting the Simons Institute for the Theory of Computing. We also thank the anonymous reviewers for their careful reading of our manuscript and their many insightful comments and suggestions. 

\section*{Funding}
R. Jin and R. Ward are supported by AFOSR MURI Award N00014-17-S-F006. R. Ward is supported by the National Science Foundation under Grant No. DMS-1638352.  T. G. Kolda is supported by the U.S. Department of Energy, Office of Science, Office of Advanced Scientific Computing Research, Applied Mathematics program. Sandia National Laboratories is a multimission laboratory manage and operated by National Technology and Engineering Solutions of Sandia, LLC., a wholly owned subsidiary of Honeywell International, Inc., for the U.S. Department of Energy's National Nuclear Security Administration under contract DE-NA-0003525.

\bibliographystyle{plain}
\bibliography{references}

\appendix
\section{Proof of \cref{modified}}
\label{modified proof}
\begin{proof}
Due to the $(2s, \delta)$-RIP property of $\vPsi$, apply the result from \cref{RIP-norm 2}, for any row and column index sets $I, J \subset [n]$, $\|\vPsi_L(I)^\top \vPsi_R(J)\| \leq \delta$, if $|I|\leq s, |J|\leq s$.

\begin{align*}
&\|\vC_{\vxx, \vy}\| = \sup_{\|\vu\|_2=\|\vv\|_2 =1}|\langle \vu, \vC_{\vxx, \vy}\vv\rangle|\\
& \leq \sup_{\|\vu\|_2=\|\vv\|_2 =1}\sum_{p, q=2\atop p \neq q}^r |\langle \vu(I_p), \vC_{\vxx, \vy}(I_p, J_q)\vv(J_q)\rangle|\\
& \leq \sup_{\|\vu\|_2=\|\vv\|_2 =1}\sum_{p, q=2\atop p \neq q}^r \|\vu(I_p)\|_2\,\|\vv(J_q)\|_2\, \|\vD_{\vxx(I_p)}\vPsi_L(:, I_p)^\top \vPsi_R(:, J_q)\vD_{\vy(J_q)}\|\\
& \leq \sup_{\|\vu\|_2=\|\vv\|_2 =1}\sum_{p, q=2\atop p \neq q}^r
\|\vu(I_p)\|_2\,\|\vv(J_q)\|_2\,\|\vxx(I_p)\|_\infty \, \|\vy(J_q)\|_\infty\, \delta\\
& \leq \delta\sup_{\|\vu\|_2=\|\vv\|_2 =1}\sum_{p, q=2\atop p \neq q}^r\|\vu(I_p)\|_2\,\|\vv(J_q)\|_2\,\frac{1}{\sqrt{s}}\, \|\vxx(I_{p-1})\|_2\, \frac{1}{\sqrt{s}}\, \|\vy(J_{q-1})\|_2\\
& \leq \frac{\delta}{s}\, \|\vxx\|_2\, \|\vy\|_2 \sup_{\|\vu\|_2=\|\vv\|_2 =1}\sum_{p, q=2\atop p \neq q}^r \left(\frac{1}{2}\|\vu(I_p)\|_2^2 +\frac{1}{2}\frac{\|\vxx(I_{p-1})\|_2^2}{\|\vxx\|_2^2}\right) \left(\frac{1}{2}\|\vv(J_q)\|_2^2 +\frac{1}{2}\frac{\|\vy(J_{q-1})\|_2^2}{\|\vy\|_2^2}\right)\\
& \leq \frac{\delta}{s}\, \|\vxx\|_2\, \|\vy\|_2.
\end{align*}

\begin{align*}
\|\vC_{\vxx, \vy}\|_F^2& = \sum_{i \not\sim j \atop i\in I_1^c, j\in J_1^c}^{n}\left(x(i)\vPsi_L(:, i)^\top\vPsi_R(:, j) y(j)\right)^2\\
& = \sum_{q=2}^r \sum_{i \in I_1^c\atop i \not\in J_q}^{n} x(i)^2\|\vD_{\vy(J_q)}\vPsi_R(:, J_q)^\top \vPsi_L(:, i)\|^2\\
& \leq \sum_{q=2}^r \sum_{i \in I_1^c\atop i \not\in J_q}^{n}x(i)^2\, \|\vy(J_q)\|_\infty^2\, \|\vPsi_R(:, J_q)^\top \vPsi_L(:, i)\|^2\\
& \leq \sum_{q=2}^r\frac{\delta^2}{s}\, \|\vy(J_{q-1})\|_2^2\, \sum_{i = 1}^n x(i)^2\\
& \leq \frac{\delta^2}{s}\, \|\vxx\|_2^2 \, \|\vy\|_2^2.
\end{align*}

\begin{align*}
\|\vv_{\vxx, \vy}\|_2&\leq \sup_{\|\vu\|_2 =1}\sum_{p=2}^r \langle \vu(I_p), \vD_{\vxx(I_p)}\vPsi_L(:, I_p)^\top\vPsi_R(:, J_1)\vD_{\vbbb}\vy(J_1)\rangle\\
& \leq \sup_{\|\vu\|_2 =1}\sum_{p=2}^r \|\vu(I_p)\|_2 \, \|\vxx(I_p)\|_\infty\,\|\vbbb\|_\infty\, \|\vPsi_L(:, I_p)^\top\vPsi_R(:, J_1)\| \,\|\vy(J_1)\|_2\\
& \leq \sup_{\|\vu\|_2 =1}\sum_{p=2}^r \|\vu(I_p)\|_2 \, \frac{1}{\sqrt{s}}\,\|\vxx(I_{p-1})\|_2\, \delta\, \|\vy\|_2\\
& \leq\frac{\delta}{\sqrt{s}}\, \|\vxx\|_2\, \|\vy\|_2\sup_{\|\vu\|_2 =1}\sum_{p=2}^r\left(\frac{1}{2}\|\vu(I_p)\|_2^2 +\frac{1}{2}\frac{\|\vxx(I_{p-1})\|_2^2}{\|\vxx\|_2^2}\right)\\
& \leq \frac{\delta}{\sqrt{s}}\, \|\vxx\|_2\, \|\vy\|_2.
\end{align*}

\begin{align*}
|w_{\vxx, \vy}| & \leq \sum_{p=1}^r \Big|\vdd(I_p)^\top \vD_{\vxx(I_p)}\vPsi_L(:, I_p)^\top \vPsi_R(:, J_p)\vD_{\vy(J_p)}\vdd(J_p)\Big|\\
& \leq \sum_{p=1}^r \|\vxx(I_p)\|_2\, \|\vy(J_p)\|_2\, \|\vPsi_L(:, I_p)^\top \vPsi_R(:, J_p)\|,\\
& \leq \delta\sum_{p=1}^r\|\vxx(I_p)\|_2\, \|\vy(J_p)\|_2\\
& \leq \delta \, \|\vxx\|_2\, \|\vy\|_2\sum_{p=1}^r \frac{1}{2}\left(\frac{\|\vxx(I_p)\|_2^2}{\|\vxx\|_2^2}+\frac{\|\vy(J_p)\|_2^2}{\|\vy\|_2^2}\right)\\
& = \delta \, \|\vxx\|_2\, \|\vy\|_2.
\end{align*}

\end{proof}

\section{Fitting CP model with alternating randomized least squares}
\label{CPRAND-MIX}
In this section, we give supplemental material on the CPRAND-MIX algorithm and show that the application of KFJLT to the alternating least squares problem greatly reduces the workload of CP tensor decomposition. 

The Khatri-Rao product, also called the column-wise Kronecker product denoted by $\bigodot$, is defined as: given matrices $\vx\in\mathbb{R}^{m\times n}$ and $\vY\in\mathbb{R}^{m'\times n}$,
\begin{equation}
\vx \odot \vY= \left[                 
  \begin{array}{cccc}   
    \vxx_{1}\otimes\vy_{1}&\vxx_{2}\otimes\vy_{2}&\dots&\vxx_n\otimes\vy_n  
  \end{array}
\right]\in\mathbb{R}^{mm'\times n}.
\end{equation}

The Khatri-Rao product also satisfies the distributive property
\begin{equation}
\vW\vx\odot \vY\vZ = (\vW\otimes \vY)(\vx\odot \vZ).
\end{equation}

\subsection{Problem set-up}
Let $\vX : \Real^{n_1} \times \dots \times \Real^{n_d}$ be a $d$-way tensor, and $\vm$ be a low-rank approximation of $\vX$ such that rank$(\vm) \leq R$. $\vm$ is defined by $d$ factor matrices, i.e. $\vA_k \in \Real^{n_k \times R}$ via
\begin{equation}
\label{low rank}
\vm = \sum_{j=1}^{R} \vA_1(:,j) \circ \vA_2(:,j) \circ \dots \circ \vA_d(:,j).
\end{equation}

The goal of fitting tensor CP model is to find the factor matrices that minimize the nonlinear least squares objective:
\begin{equation}
\label{total ls}
\| \vX - \vm \|^2 = \sum_{i_1 = 1}^{n_1}\sum_{i_2 = 1}^{n_2}\dots \sum_{i_d = 1}^{n_d} (x(i_1, i_2, \dots, i_d) - m(i_1, i_2, \dots, i_d))^2
\end{equation}
subject to $\vm$ being low rank as \cref{low rank}. The tensor $\vX$ has $\prod_{k=1}^{d} n_k$ parameters whereas $\vm$ has only $R\sum_{k=1}^{d} n_k$ parameters.

For ease of the notation, we define $N = \prod_{k=1}^{d}n_k$ and $N_k = N/n_k$. 

The mode-$k$ unfolding of $\vX$ recognizes the elements of the tensor into a matrix $\vx_{(k)}$ of size $n_k \times \vN_k$. The mode-$k$ unfolding of $\vm$ has a special structure:
\begin{equation}
\label{flatten m}
\vm_{(k)} = \vA_k \underbrace{(\vA_d \odot \dots \odot \vA_{k+1} \odot \vA_{k-1} \dots \odot \vA_1)^\top}_{\vZ_k^\top}.
\end{equation}

The idea behind \emph{alternating least squares} (ALS) for CP is solving for one factor matrix $\vA_k$ at a time, repeating the cycle until the method converges. This takes advantage of the fact that we can rewrite the minimization problem using \cref{flatten m} as
\begin{equation}
\label{explicit cpals}
\min_{\vA_k} \|\vZ_k\vA_k^\top - \vx_{(k)}^\top\|_{F},
\end{equation}
which is a linear least square problem with a closed form solution. The cost of solving the least square problem is $O(R N)$ due to the particular structure of $\vZ_k$. But we need to solve $d$ such problems per outer loop and run tens or hundreds of outer loops to solve a typical CP-ALS problem. Hence, reducing the cost of \cref{explicit cpals} is of interest.

\subsection{Randomized least squares}
Since we expect that the number of rows $\vN_k$ is much greater than the number of columns $R$ in $\vZ_k$, \cref{explicit cpals} can benefit from randomized sketching methods. Instead of solving the full least square, we can instead solve a reduced problem by a sketch matrix $\vPhi\in\mathbb{C}^{m \times \vN_k}$:
\begin{equation}
\min_{\vA_k} \|\vPhi\vZ_k\vA_k^\top - \vPhi\vx_{(k)}^\top\|.
\end{equation}

For $\vPhi$ being a FJLT, the dominant cost are applying the FFT to $\vZ_k$ and $\vx_{(k)}$ and solving the least square: $O((R+n_k) \vN_k \log \vN_k + R^2m_{\text{f}})$.

We then change the sketching form of $\vPhi$ to be a KFJLT: 
\begin{equation}
\vPhi = \vS\left(\bigotimes_{\ell = d \atop \ell \neq k}^1 \vF_{\ell} \vD_{\ell} \right),
\end{equation}
as we can compute \cref{explicit cpals}  more efficiently.

First consider the multiplication with $\vx_{(k)}^\top$. We pay an \emph{one-time} upfront cost to reduce the cost per iteration. The corresponding computation is to mix the original tensor:
\begin{equation}
\hat{\vX} = \vX \times_1 \vF_{1} \vD_{1} \dots \times_d \vF_{d} \vD_{d}.
\end{equation}
The total cost is $N \log N$.

We observe that
\begin{equation}
\vPhi\vx_{(k)}^\top = \left(\vS \hat{\vx}_{(k)}^\top\right)\vF_{k}^{\ast} \vD_{k}.
\end{equation}
The asterisk $\ast$ denotes the conjugate transpose. This equation shows that we just need to sample and then apply the inverse FFT and diagonal. The work per iteration is $O(m_{\text{kron}} n_k \log n_k)$.

Next consider $\vPhi\vZ_k$. We finish the mixing for $\vA_k$: $\hat{\vA}_k = \vF_{k} \vD_{k} \vA_k$, which costs\\
 $O(R n_k \log n_k)$, before sketching the least square in mode $k$. Then the cost of computing
\begin{equation}
\vPhi \vZ_k = \vS \left(\bigodot_{\ell = d \atop \ell\neq k}^1 \vF_{\ell}\vD_{\ell}\vA_{\ell}\right) = \vS\left(\bigodot_{\ell = d \atop \ell\neq k}^1 \hat{\vA}_{\ell}\right)
\end{equation}
is just the cost of sampling the Khatri-Rao product: $Rm_{\text{kron}}$.  

To conclude the comparison of the cost in \cref{embed cost}:
\begin{table}[tbhp]
{\footnotesize
  \caption{\bf Cost per inner iteration}
\label{embed cost}
\begin{center}
  \begin{tabular}{cccc} \toprule[1pt]
\bf Regular CP-ALS & \bf FJLT~-~sketched & \bf Kronecker FJLT~-~sketched \\ \hline
$\bm{O(RN)}$ &  $\bm{O((R+n_k) \vN_k \log(\vN_k) +R^2m_{\text{f}})}$ & $\bm{O((R+m_{\text{kron}}) n_k \log n_k +R^2m_{\text{kron}})}$\\
\bottomrule[1pt] 
\end{tabular}
\end{center}
}
\end{table}

It is natural to choose KFJLT as the sketch strategy for solving CP alternating least squares, as it helps reduce the cost of the inner iteration greatly to the order of $O(n_k \log n_k)$ compared to the original cost: $O(N)$. This idea has been developed into a randomized algorithm: CPRAND-MIX. We refer the readers to \cite{BBK18} for the completed algorithm. 
\end{document}